\documentclass[11pt,letterpaper]{article}

\usepackage[margin=1in]{geometry}
\usepackage{array}
\usepackage{graphicx}
\usepackage{multicol}
\usepackage{multirow}
\usepackage{pgfplots}
\usepackage{tikz}
\usepackage[version=4]{mhchem}
\usepackage{amsmath,amssymb,amsfonts,mathtools}
\usepackage{appendix}
\usepackage{bbm}
\usepackage{enumitem}
\usepackage{amsthm}
\usepackage{titlesec}
\usepackage{graphicx}
\usepackage{textcomp}
\usepackage{xcolor}
\usepackage{subcaption}
\usepackage{algpseudocode}
\usepackage{algorithm}
\usepackage{verbatim}
\usepackage{setspace}
\usepackage{tablefootnote}
\usepackage{fancyhdr}
\usepackage{thmtools}
\usepackage{thm-restate}
\pgfplotsset{compat=1.18}

\newcommand{\mc}{\mathcal}
\newcommand{\br}[1]{\left\{#1\right\}}

\usepackage[linktoc=all]{hyperref}
\hypersetup{
    colorlinks=true,
    linkcolor=black,
    filecolor=blue,      
    urlcolor=blue,
    citecolor=blue,
    pdfpagemode=FullScreen,
    }

\usepackage[capitalise, noabbrev]{cleveref}

\usepackage[sortcites, natbib=true, url=false, style=alphabetic, doi=false, isbn=false, maxbibnames=9, minalphanames=3]{biblatex}

\renewcommand{\l}{\left (}
\renewcommand{\r}{\right)}
\newtheorem{theorem}{Theorem}[section]
\newtheorem{lemma}[theorem]{Lemma}
\newtheorem{definition}[theorem]{Definition}
\newtheorem{example}[theorem]{Example}
\newtheorem{remark}[theorem]{Remark}
\newtheorem{claim}[theorem]{Claim}
\newtheorem{corollary}[theorem]{Corollary}

\newcommand{\E}{\mathbb E}
\renewcommand{\P}{\mathbb P}

\newcommand{\Z}{\mathbb Z}
\newcommand{\poly}{\operatorname{poly}}
\newcommand{\polylog}{\operatorname{polylog}}
\newcommand{\Alg}{\textsc{Alg}}
\newcommand{\R}{\mathbb R}
\DeclareMathOperator{\lca}{lca}

\newcommand{\seq}{\boldsymbol}
\begin{document}
\title {Online Monotone Metric Embeddings}
\author{Christian Coester\\ University of Oxford \and Yichen Huang\\ Harvard University}
\date{}
\maketitle
\titlespacing{\paragraph}{%
  0pt}{
  0.5\baselineskip}{
  1em}
\begin{abstract}
\thispagestyle{plain}

Metric embeddings into structured spaces, particularly hierarchically well-separated trees (HSTs), are a fundamental tool in the design of online algorithms. In the classical online embedding setting, points arrive sequentially and must be embedded irrevocably upon arrival, resulting in strong distortion lower bounds of $\Omega(\min(n, \log n\log \Delta))$, where $n$ is the number of points and $\Delta$ their aspect ratio.

We propose a novel relaxation, \emph{online monotone metric embeddings}, which allows distances between embedded points in the target space to decrease monotonically over time. Such relaxed embeddings remain compatible with many online algorithms. Moreover, this relaxation breaks existing lower bound barriers, enabling embeddings into HSTs with distortion $O(\log^2 n)$.

We also study a dynamic variant, where points may both arrive and depart, seeking distortion guarantees in terms of the maximum number $l$ of simultaneously present points. For traditional embeddings, such bounds are impossible, and this limitation persists even for deterministic monotone embeddings. Surprisingly, probabilistic monotone embeddings allow for $O(l \log l)$ distortion, which is nearly optimal given an $\Omega(l)$ lower bound.

\end{abstract}

\pagenumbering{gobble}
\newpage
\pagenumbering{arabic}
\section{Introduction}

Metric embeddings are a powerful tool in many online and approximation algorithms. By embedding a metric space into a structured one with small distortion, one can solve algorithmic problems on the simpler metric and then translate solutions back to the original space with bounded loss. Here we focus on embeddings into \emph{hierarchically well-separated trees (HSTs)}, a widely studied target metric space with numerous algorithmic applications~\cite{Bartal96, FRT03}.

When applying metric embeddings in the context of online algorithms, the traditional approach embeds the entire underlying metric into an HST offline, subsequently solving the problem there~\cite{Bartal2020, Bartal96, FRT03}. However, the distortion of these embeddings inherently depends on the size of the metric space $M$ (e.g., at least $\Omega(\log |M|)$ \cite{Bartal96}), which can be infinite even for simple metrics like the real line. This motivates the study of \emph{online embeddings}, which only embed points relevant to the request sequence and thus involve only finitely many points.

\paragraph{Online Embeddings and Their Challenges.} 
In an online embedding, points arrive sequentially. Upon arrival of each new point $v$, the embedding algorithm must irrevocably place $v$ into the target metric without knowledge of future points, and it cannot later move points in this embedding under traditional models~\cite{Indyk10online, Bartal2020, Bhore24duet}. This restriction leads to strong distortion lower bounds: a deterministic lower bound of $2^{\Omega(n)}$ \cite{newman2020online} and a randomized lower bound of $\widetilde{\Omega}(\min (n, \log n \log \Delta))$~\cite{Indyk10online, Bartal2020}, where $n$ is the length of the sequence and $\Delta$ the aspect ratio, i.e., the ratio between the maximum and minimum nonzero distances between points. Since $\Delta$ can be unbounded, these bounds are exponentially larger than those in the offline setting, which admits deterministic distortion $O(n)$ and randomized distortion $O(\log n)$.

\paragraph{Monotone Metric Embeddings: A Mild but Powerful Relaxation.}
Our starting observation is that requiring a fully fixed embedding at each step is often unnecessarily restrictive. In many algorithmic settings, permitting distances between already-embedded points to \emph{decrease} over time is a benign change that would not degrade performance. For example, with a \emph{non-contractive} embedding that never underestimates distances, decreasing distances among existing points only improves the approximation quality. As a consequence, many online metric algorithms remain competitive even when paired with an embedding that occasionally shrinks distances, and we provide a meta-theorem for verifying the compatibility of potential-based algorithms.

This motivates a new model that we call \emph{online monotone metric embeddings}: upon arrival of a new point, previously embedded points may be repositioned, provided that no pairwise distance between already-embedded points increases. This additional flexibility turns out to evade the lower bounds known for the traditional model. Specifically, we achieve probabilistic online monotone embeddings into HSTs with distortion $O(\log^2 n)$, eliminating the dependence on $\Delta$ and approaching the offline embedding distortion of $\Theta(\log n)$~\cite{FRT03}. We also obtain tight $\Theta(n)$ bounds on \emph{deterministic} online monotone embeddings, improving exponentially over the strict\footnote{We refer to the traditional embedding model, where distances in the target metric must be fixed upon arrival, as the \emph{strict} embedding model.} setting.

\paragraph{Dynamic Embeddings with Removals.}
Another natural question is whether an online embedding must always maintain an embedding of \emph{all} points ever introduced. In practice, only a small set of \emph{alive} points may matter at any given time. For example, in metrical service systems (i.e., set chasing problems)~\cite{bubeck22shortest, PAPADIMITRIOU1991127,CoesterT26}, only the points in the current and previous request sets matter. In $k$-server or $k$-taxi problems, points requested long ago intuitively lose relevance over time. Motivated by such scenarios, we examine a \emph{dynamic} embedding model where points may arrive and depart, such that at most $l$ points are alive (arrived and have not departed) at any time. The algorithm needs only to maintain an embedding of the alive set.

Previous literature on the $k$-server problem has attempted to employ such an approach, with partial success: \cite{Bubeck18} achieved a $\polylog(k,\Delta)$-competitive algorithm for the $k$-server problem using an evolving HST embedding. Building upon this, \cite{Lee2018FusibleHA} aimed to further reduce the competitive ratio to $\polylog(k)$ by refining the dynamic embedding, and while this work contained several promising ideas, it was later withdrawn due to a bug in the proof. Despite this interest, a systematic exploration of dynamic embedding models remains limited. We investigate what distortion as a function of $l$, if any, is achievable in the dynamic embedding setting.

In the strict setting, known results on embedding graphs of pathwidth $l$ \cite{Lee2009PathwidthTA} imply an \emph{offline} embedding with distortion $O((4l)^{l^3+1})$ into trees. For \emph{strict online} embeddings into HSTs, even such exponential distortion is unachievable: for $l=3$, the distortion is an unbounded function of $n$. For online monotone embeddings, the same impossibility still holds in the case of deterministic embeddings. Surprisingly, we show that the combination of monotonicity and randomization removes these barriers, enabling a distortion of $O(l \log l)$. We further establish an almost matching lower bound of $\Omega(l)$ for probabilistic embeddings, which holds even offline. This lower bound may be of independent interest, particularly for future research seeking to narrow the gap in the randomized competitive ratio of the $k$-server problem. It could inspire new lower bound constructions, or at least constrain possible approaches when aiming for a $\polylog(k)$-competitive algorithm.

\medskip

In summary, monotone metric embeddings offer a natural relaxation of the online embedding model that enables overcoming known lower bounds, often exponentially. To the best of our knowledge, this notion of ``monotone recourse'' is new and may have broader applicability to other online problems. We view this conceptual idea as our main contribution. 
Some of our results follow from relatively simple adaptations of existing techniques, illustrating that monotone recourse can be addressed within known algorithmic frameworks.

\subsection{Our Results}

We will refer to the notion of an update sequence, formally defined in \cref{def: update sequence}, which specifies the arrival and departure of points. We denote by $n$ the total number of points in the sequence, and by $l$ the \emph{width}, the maximum number of points simultaneously present at any time.

Our embedding results are presented in two categories: the \emph{incremental} setting\footnote{
Although our incremental algorithms can accommodate departures, when focusing on an $n$-dependent bound, having departures only simplifies the task by reducing the number of alive points.}, in which points only arrive and the distortion is analyzed as a function of $n$, and the \emph{fully dynamic} setting, in which points arrive and depart and the distortion is analyzed as a function of $l$. Finally, we discuss some applications.

\subsubsection{Incremental Setting}
We first present improved probabilistic embeddings for the classical incremental model when monotone updates are allowed.

\begin{restatable}[Probabilistic Embedding]{theorem}{Incremental}
\label{thm: general}
For every $n \in \mathbb N$, there exists a probabilistic online monotone embedding of up to $n$ points from any metric space into HSTs with distortion $O(\log^2 n)$. If the points are from an $O(1)$-dimensional normed space, the distortion improves to $O\bigl(\log n\bigr)$. 
\end{restatable}

Although \cref{thm: general} requires prior knowledge of $n$, we can employ it in our applications (Theorems~\ref{thm: k-server} and \ref{thm: k-taxi}) even in situations where $n$ is unknown by using a standard guess-and-double approach. For the pure embedding question with unknown $n$, a modified algorithm gives distortion $O(\log^2 n \log \log n)$. Note that the $\Omega(\log n)$ lower bound for offline embeddings extends to our setting, so our bounds for constant-dimensional normed spaces are tight, whereas a quadratic gap remains for general metrics.

We also analyze deterministic online monotone embeddings, obtaining tight bounds. Recall that the optimal distortion for deterministic strict online embeddings is exponential in $n$.

\begin{restatable}[Deterministic Embedding]{theorem}{Deterministic}
\label{thm: deterministic}
For every $n \in \mathbb N$, there is a deterministic online monotone embedding of up to $n$ points from any metric space into HSTs with distortion $n-1$ when $n$ is known. Without prior knowledge of $n$, the distortion is $n \cdot \widetilde{\Theta}(\log n)$.\footnote{$\widetilde \Theta$ hides $\log \log n$ factors.} Both bounds are tight.
\end{restatable}

\subsubsection{Fully Dynamic Setting}
In the fully dynamic model, points may both arrive and depart. We first present negative results demonstrating that strict dynamic embeddings and deterministic monotone embeddings cannot achieve bounded distortion purely in terms of the width $l$.

\begin{restatable}[Dynamic Impossibility]{theorem}{DynamicLowerDeter}
\label{thm: lower bound dynamic}
Even for update sequences of width $3$, any deterministic strict dynamic embedding into HSTs incurs distortion $\Omega(2^n)$, and any probabilistic strict embedding incurs distortion $\Omega(n)$. Furthermore, any deterministic monotone embedding into HSTs incurs distortion $\Omega(n)$ for update sequences of width $3$.
\end{restatable}

Despite this deterministic impossibility, randomization surprisingly enables embeddings whose distortion is bounded by the width $l$:

\begin{restatable}[Probabilistic Dynamic Embedding]{theorem}{GeneralDynamic}
\label{thm: general dynamic}
There exists a probabilistic online monotone embedding from any metric space into HSTs with distortion $O(l \log l)$, where $l$ is the width of the sequence. If the points are from $O(1)$-dimensional normed space, the distortion improves to $O\bigl(l\bigr)$. These results hold without prior knowledge of $l$.
\end{restatable}

We complement this result with an almost matching lower bound of $\Omega(l)$. Our lower bound is constructed on the line and holds even offline. The proof is short but non-trivial, and may be of interest for future research aimed at closing the gap in the randomized competitive ratio for the $k$-server problem.

\begin{restatable}[Dynamic Lower Bound]{theorem}{DynamicLowerRand}
\label{thm: lower bound dynamic rand}
For every $l\in\mathbb N$, there exists an update sequence of width $l$ on the line such that any probabilistic monotone embedding of the sequence into HSTs incurs distortion $\Omega(l)$, even if the entire sequence is known in advance.
\end{restatable}

\subsubsection{Applications}

Finally, we demonstrate some applications of our embeddings for online algorithms. To give a systematic meta-theorem, we consider online algorithms whose analysis is based on a potential function. Note that the competitiveness of an online algorithm is \emph{equivalent} to the existence of a corresponding potential function~\cite{BenDavidBKTW94}.\footnote{However, it can be difficult to find an explicit expression for the potential.} Intuitively, the potential value quantifies the amount of disadvantage of the online algorithm's configuration. A natural property of potential functions, which is typically satisfied (cf. Section~\ref{sec: monotone potential}), is to be \emph{non-decreasing} in distances of the metric space. The following theorem, stated more formally in Theorem~\ref{thm: reduction}, shows that this condition suffices in order for an online algorithm to be compatible with our monotone embeddings.

\begin{theorem}[Application, Informal]\label{thm:informalApplication}
Consider an online problem $\Pi$ on a metric $N$. If there exists an online monotone embedding from $N$ into a family of metrics $\mathcal{M}$ with distortion $\lambda$, and a $\rho$-competitive algorithm for $\Pi$ on every metric $M \in \mathcal{M}$ using a potential function that is monotone non-decreasing in distances, then there is a $\rho \lambda$-competitive algorithm for $\Pi$ on $N$.
\end{theorem}

As examples, we recover a result for $k$-server in~\cite{Bartal2020}, and give new applications to the $k$-taxi problem.

\begin{restatable}[$k$-Server]{theorem}{kserver}
\label{thm: k-server}
    There is an $O(\log^2 k \log^2 n)$-competitive algorithm for the $k$-server problem on general metrics and an $O(\log^2 k \log n)$-competitive algorithm on $O(1)$-dimensional normed spaces, where $n$ is the number of requested locations.
\end{restatable}

The $O(\log^2 k\log^2 n)$ bound matches the guarantee in~\cite{Bartal2020}, where it was obtained by combining an $O(\log^2k)$-competitive HST algorithm with an online embedding of distortion $O(\log n\log\Delta)$ paired with algorithm combination techniques (applicable to problems that admit a so-called min-operator) and an additional ``baseline'' algorithm, which allows to replace $\Delta$ by $n$ in the overall competitiveness. Our algorithm offers an alternative perspective on this result, where the $O(\log^2 n)$ factor in the competitiveness comes directly from the embedding, without requiring a min-operator or a baseline algorithm.

Similarly, our embedding yields an alternative, more direct method to recover an $O(\log^2 n)$-competitive algorithm (\cref{thm: constraint forest}) for the subadditive constrained forest problem~\cite{KP95Constrained, Bartal2020}.

For the $k$-taxi problem, we obtain the following new result.

\begin{restatable}[$k$-Taxi]{theorem}{ktaxi}
\label{thm: k-taxi}
There is an $O(2^k \log^{2} n)$-competitive algorithm for the $k$-taxi problem on general metrics, and an $O(2^k \log n)$-competitive algorithm for $O(1)$-dimensional normed spaces, where $n$ is the number of requested locations.
\end{restatable}

We note that in previous guarantees for the $k$-taxi problem~\cite{coester2019online,BubeckBCS21,BuchbinderCN23,GuptaKP24}, $n$ refers instead to the total number of points in the underlying metric space (including those that are never requested), which is infinite even for simple metrics like the line; see \cref{sec: related} for details on previous results.

Our dynamic embedding with distortion $O(l \log l)$ also suggests a potential avenue for addressing the long-standing open question of whether there exists an algorithm for the $k$-taxi problem whose competitive ratio depends only on $k$. In particular, \cref{thm: k-taxi equivalent} demonstrates that it suffices to maintain online a set of $g(k)$ points such that the optimal \emph{offline} algorithm restricted to these points achieves an $h(k)$-approximation.

\subsection{Related Work}
\label{sec: related}
\paragraph{Embeddings into HSTs.}
Offline, any metric of $n$ points can be embedded into an HST with $O(\log n)$ distortion \cite{FRT03}. For the online model, \cite{Indyk10online} adapted the offline algorithm of \cite{Bartal96} to achieve $O(\log n \log \Delta)$ distortion, and \cite{Bartal2020} provides an almost matching lower bound of  $\widetilde{\Omega}(\log n \log \Delta)$ on probabilistic embeddings into HSTs. For a distortion allowed to depend on $n$ only, it becomes $\Omega(n)$~\cite{Indyk10online}. For metrics with doubling dimension $\textsc{ddim}$, \cite{Bhore24duet} achieves $O(\textsc{ddim}\,\log \Delta)$-distortion online embeddings into HSTs.

\paragraph{Online Problems and Algorithms.} 
The $k$-server problem is one of the most prominent problems in the field of online algorithms. Deterministically, the competitive ratio is $\Theta(k)$ \cite{MMS88, KP95}. The paper \cite{Bubeck18} achieved a randomized competitive ratio of $O(\log^2 k)$ on HSTs, which implies an $O(\log^2k \log n)$-competitive algorithm on any $n$-point metric. The reader is referred to \cite{KOUTSOUPIAS09kserver} for a survey on the $k$-server problem, and~\cite{BubeckCR23} for more recent results. The latter~\cite{BubeckCR23} also shows that the approach via HST embeddings is in fact optimal for the related metrical task systems problems on general metrics, despite the distortion loss.

The $k$-taxi problem is a generalization of the $k$-server problem \cite{fiat1990}. Each request in the $k$-taxi problem is a pair of points $(s, t)$, requesting a taxi to first come to $s$ and then to $t$. In the hard version of the problem, the cost is the total distance taxis travel without a customer. This hard version has an $\Omega(2^k)$ lower bound on the competitive ratio for randomized algorithms against adaptive adversaries~\cite{coester2019online} and a few (incomparable) upper bounds: $O(2^k \log n)$, $O((n \log k)^2 \log n)$, $2^{O(\sqrt{\log k \log \Delta})} \log_\Delta n$, and $O(\log^3 \Delta \log^2(nk\Delta))$~\cite{coester2019online,BubeckBCS21,BuchbinderCN23,GuptaKP24}, all of which use $n$ for the number of points in the entire metric space. On general metrics (with infinitely many points), competitive algorithms are known only for $k\le 3$~\cite{CoesterP26}.

When applying online embeddings with online algorithms, \cite{Bartal2020} explains a framework to bypass the dependence on the aspect ratio $\Delta$ at the expense of another $O(\log n)$ factor for problems that admit min operators and belong to a class they call \textit{abstract network design problem}. Our results apply to a class of \emph{metrical request-answer games} that includes abstract network design and other metrical request-answer games. 

\paragraph{Online Algorithms with Recourse.} The idea of allowing some changes to previous actions is related to the notion of recourse in online and dynamic algorithms, with the focus on bounding the \emph{number} or \emph{cost} of changes that an algorithm makes~\cite{Bhattacharya16new, Bernstein2021DynamicMatching, Nilsen17fully, Bhore24duet}. Unlike typical models with recourse, where allowing changes introduces a tradeoff or cost, our monotone updates only decrease distances and hence can be viewed as benign flexibility rather than true recourse. We therefore view this as a free monotone relaxation rather than a standard recourse model.

\paragraph{Dynamic Algorithms.} The field of \emph{dynamic algorithms} also studies problems whose input is a sequence of arrivals and departures (e.g., of nodes or edges in a graph) while a solution to the current instance has to be maintained \cite{Bhattacharya16new, Bernstein2021DynamicMatching,Nilsen17fully,Eppstein97, frederickson91}. The goal in these problems is typically to minimize update time when computing new solutions, without restrictions on the type of changes allowed when updating the solution. In contrast, we allow only monotone changes but do not impose any restrictions on the running time, as is common in the field of online algorithms.

\subsection{Organization}
\cref{sec: prelim} defines the notions and the models and provides some useful lemmas. \cref{sec: overview} explains the main technical ideas and a general framework for our algorithms. \cref{sec: online} proves the main results for the incremental setting (distortion as a function of $n$), and \cref{sec: dynamic} for the fully dynamic setting (distortion as a function of $l$). \cref{sec: application} introduces the metrical request-answer game and proves applications for our embedding. \cref{sec: conclusion} contains concluding remarks and highlights future directions. Omitted proof details, embeddings from normed spaces, and discussions for deterministic monotone embeddings are presented in the appendix.

\section{Preliminaries}
\label{sec: prelim}

\subsection{Basic Notions}
\label{sec: notion}

\newcommand{\D}{\mathcal D}

For a set $S$, let $\D(S)$ be the set of probability distributions over $S$. We use $\P(\cdot)$ to denote probabilities and $\E[\cdot]$ for expectations. For a sequence $a \in A^*$, the notation $a_{[i,j]}$ denotes the subsequence $a_i, \ldots, a_j$.

\begin{definition}[HSTs]
    For $\mu \ge 1$, a $\mu$-\emph{hierarchically well-separated tree} (HST) is a metric space whose points are the leaves of a rooted tree $T$. Each node\footnote{The original definition in \cite{Bartal96} assigns weights to edges instead of nodes; both formulations are equivalent up to a constant factor (see, e.g., \cite{bartal2021advances, Bartal2020}).} $v$ of $T$ has a weight $\varphi(v) \geq 0$, with $\varphi(v) = 0$ if and only if $v$ is a leaf, and if $v$ is a child of $u$, then $\varphi(v) \le \varphi(u)/\mu$. The distance between two leaves $u$ and $v$ is given by $d_T(u, v) = \varphi(\lca(u, v))$, where $\lca(u, v)$ denotes the least common ancestor of $u$ and $v$.
\end{definition}

Throughout the paper, we refer to an original metric space $(X, d_X)$ and consider embeddings of finite subsets $V \subseteq X$ into HSTs. Our probabilistic embedding will fix $\mu=2$. Unless specified otherwise, $d(u,v)$ (without a subscript) refers to the distance in $X$. 
We write $d_{\max}(V) = \max_{u,v \in V} d(u,v)$ for the diameter of $V$, and $d_{\min}(V) = \min_{u \neq v \in V} d(u,v)$ for the smallest nonzero distance in $V$. The \emph{aspect ratio} of $V$ is $\Delta(V) = \frac{d_{\max}(V)}{d_{\min}(V)}$. We may omit the argument $V$ and write $\Delta$ when it is clear from context. For two metric spaces $M_1 = (V_1, d_1)$ and $M_2 = (V_2, d_2)$, we say $M_1$ \emph{dominates} $M_2$ if for all $u, v \in V_1\cap V_2$, $d_1(u, v) \geq d_2(u, v)$.

\subsection{Monotone Embeddings}

We first define an \emph{update sequence}, which captures both arrivals and departures of points:

\begin{definition}[Update Sequence]
\label{def: update sequence}
    An update sequence $\sigma$ on a metric $(X, d_X)$ is a sequence of pairs $(v_t, o_t) \in X \times \{+, -\}$. We say the point $v_t$ \emph{arrives} at time $t$ if $o_t = +$, and it \emph{leaves} at time $t$ if $o_t=-$. Let $L_t$ be the set of \emph{alive} points at time $t$: We set $L_0 = \varnothing$, and then for each $t \ge 1$:
    \[
    L_t = 
    \begin{cases}
        L_{t-1} \cup \{v_t\}, & \text{if } o_t = +,\\
        L_{t-1} \setminus \{v_t\}, & \text{if } o_t = -.
    \end{cases}
    \]
    We let $n$ denote the length of the sequence and $l$ denote its width, i.e., $\max_{t} |L_t|$. We write $V_t = \{v_j: j \le t, o_j = +\}$ for the set of all points introduced up to time $t$.
\end{definition}

Traditional metric embeddings map to a single fixed target metric. In our framework, we allow the target metric to change over time, provided distances between previously embedded points never increase. Thus, rather than mapping to one fixed metric, we map to a \emph{family} of metrics (in our case, HSTs) in an online manner. Similar to~\cite{Bartal2020}, we focus on non-contractive embeddings, so we omit the term ``non-contractive'' here, and just call them embeddings. 

\begin{definition}[Online Monotone Embedding]
\label{def: embedding}
Let $(X, d_X)$ be a metric space and let $\mathcal{M}$ be a family of metric spaces. A \emph{(deterministic) online monotone embedding} from $(X, d_X)$ into $\mathcal{M}$ takes inputs from an update sequence $\sigma$ of length $n$ one by one, and upon receiving $\sigma_t$, outputs a metric $d_t$ satisfying the following conditions:
\begin{enumerate}[itemsep=0em]
    \item $M_t = (L_t, d_t) \in \mathcal{M}$;
    \item $d_t$ dominates $d_X$ on $L_t$: for all $u, v \in L_t$, $d_X(u, v) \le d_t(u, v)$\label{it:noncontractive};
    \item $d_{t}$ is dominated by $d_{t-1}$ on $L_{t-1} \cap L_t$: for all $u,v \in L_{t-1} \cap L_t$, 
    $d_{t}(u,v) \le d_{t-1}(u,v)$.
\end{enumerate}

A probabilistic online monotone embedding is a probability distribution over deterministic ones. Such an embedding has \emph{distortion} $\lambda$ if, for every update sequence $\sigma$, every $t$, and $u, v \in L_t$,
\[
\E[d_t(u, v)] \le \lambda d_X(u, v),
\]
where the expectation is taken over the internal randomness of the embedding.
\end{definition}
For simplicity, the embeddings defined here are non-contractive (by property \ref{it:noncontractive}). In Appendix~\ref{app:contractions} we consider a slightly more general definition that drops this requirement. For probabilistic embeddings, we consider the \emph{oblivious adversary} model, i.e., the update sequence does not depend on the outcome of random choices made by the algorithm.

\subsection{From Partitions to HST Embeddings}
Instead of constructing an HST directly, we follow the established approach in \cite{Bartal96, Bartal2020} of building it implicitly via partitions.

\begin{definition}[Partition]
    A \emph{partition} $P$ of a set of points $V$ is a collection of disjoint subsets (called \emph{clusters}) whose union equals $V$. For a point $v \in V$, we denote by $P(v)$ the cluster containing $v$. A partition of a metric space is \emph{$s$-bounded} if each cluster has a diameter at most $s$. A probabilistic $s$-bounded partition is a distribution over $s$-bounded partitions.
\end{definition}

\begin{definition}[Smoothness Parameters]
    A probabilistic $s$-bounded partition $P$ of $V$ is said to be $(\delta,\varepsilon,\gamma)$-\emph{smooth} if, for all $u,v \in V$,
    \begin{align}
    d(u,v) \le s &\implies
    \P[P(u) \neq P(v)] \le \frac{d(u,v)}{s}\delta,\label{eq:smooth}\\
    d(u,v) \le \varepsilon s &\implies 
    \P[P(u) \neq P(v)] \le \frac{d(u,v)}{s}\delta\gamma.\label{eq:enforcing}
    \end{align}
\end{definition}

Condition \eqref{eq:smooth} resembles the padding parameter $\delta$ used in prior literature \cite{abraham2006advances, Bartal2020} but is slightly more relaxed. Condition \eqref{eq:enforcing} generalizes the notion of $\varepsilon$-enforcing from \cite{Bartal96}, which required no splitting at all when $d(u,v)\le \varepsilon s$ (i.e.\ $\gamma=0$). In the online setting, we cannot strictly enforce this; however, we will ensure a small $\gamma$ so that the distortion remains unaffected up to constant factors. We write $\delta$-smooth as an abbreviation of $(\delta, 1,1)$-smooth, in which case we only make use of property \eqref{eq:smooth}.

\begin{definition}[Online Monotone Partitions]
\label{def: online partition with recourse}
    An \emph{online monotone partition} on $(X,d_X)$ takes an update sequence $\sigma$ of length $n$ and, for each $1 \le t \le n$, generates a partition $P_t$ over $L_t$ such that:
    \begin{enumerate}
        \item $P_t$ does not depend on $\sigma_{[t+1,n]}$.
        \item For any $u,w \in L_t \cap L_{t+1}$, if $P_t(u) = P_t(w)$, then $P_{t+1}(u) = P_{t+1}(w)$.
    \end{enumerate}
    This ensures that clusters cannot be split over time: points in the cluster at time $t$ remain in the same cluster at time $t+1$. (Note that two clusters may merge, and new arrivals can join existing clusters, but splitting a previously formed cluster is disallowed.)
\end{definition}

The following lemma relates the distortion of the embedding to the smoothness of the partitions. A similar lemma that does not allow updates and $\gamma=0$ was proved in \cite{Bartal96} (see also \cite{Bartal2020}), so we provide a sketch here and defer the complete proof to \cref{appendix: tree building}. With this lemma, our main focus will be on building online probabilistic partitions with desired smoothness parameters.

\begin{restatable}[HST Construction]{lemma}{Construction}
\label{lemma: hierarchy}
    Suppose that for every integer $j$, there is a probabilistic online $2^j$-bounded monotone partition on $(X,d_X)$ for an update sequence $\sigma$ of length $n$ and that for each time $t \le n$, the probabilistic partition $C_t$ is $(\delta_t^{(j)}, \varepsilon, \gamma)$-smooth. Then there is an online monotone embedding from $(X,d_X)$ into HSTs that, on the same input, achieves a distortion of 
    \[
    O\left(\max_{t \le n}\left(\max_{j \in \mathbb{Z}} \delta_t^{(j)} \log \varepsilon^{-1} + \gamma \sum_{j \in \mathbb{Z}} \delta_t^{(j)}\right)\right).
    \]
\end{restatable}
\begin{proof}[Proof Sketch]
    The HST is constructed by constructing the partition for each scale. Each cluster\footnote{More precisely, clusters of a refined partition where higher-level partitions subdivide the lower-level clusters.} at scale $2^j$ corresponds to a node of the HST at the same scale. Because no cluster can split over time, the lowest common ancestor between two points will not move to high scales; therefore, the distances between points in the same cluster will not increase and satisfy monotonicity.

    We now bound the distortion. Fix any pair $u, v$. At a level $j$ such that $\varepsilon 2^j \le d(u,v) \le 2^j$, the probability that $u,v$ are split at level $j$ is $\delta_t^{(j)}d(u,v)/2^j$. Splitting incurs a distance of $2^j$ in the HST, contributing an expected distance of $\delta_t^{(j)} d(u,v)$. There are only $\log (\varepsilon^{-1})$ such levels, which contributes the $O(\max_j \delta_t^{(j)} \log \varepsilon^{-1})$. Remaining levels either have $s_j < d(u, v)$ and contribute a constant factor or have $d(u, v) \leq \varepsilon 2^j$ and contribute $\gamma \delta_t^{(j)} \cdot d(u, v)$. 
\end{proof}

Finally, we introduce the notion of \emph{relevant} scales. We say that a scale (or level) $j$ is \emph{relevant} at time $t$ if the partition $P_t$ is not $0$-smooth at scale $2^j$. We will use this to provide an upper bound of $\sum \delta_t^{(j)}$. We note a technical lemma for bounding the number of relevant scales, whose proof is also in \cref{appendix: tree building}.
\begin{lemma}
\label{lemma: relevant scales}
    Let $V$ be a set of $n$ points in a metric space $(X, d)$. Fix $0 < \varepsilon < 1$ and define 
    \[
    S = \{j \in \mathbb{Z} \mid \exists u,v \in V: d(u,v)\in[\varepsilon2^{j-1},2^j)\}.
    \]
    Then $|S| = O(n\log(1/\varepsilon))$.
\end{lemma}

\section{Technical Overview}
\label{sec: overview}

\paragraph{Starting Observations.} Equipped with \cref{lemma: hierarchy}, our goal is to maintain probabilistic online $s$-bounded smooth partitions at a given scale $s$. We begin by recalling some techniques used in offline embedding for $(\delta, \varepsilon,\gamma)$-smooth partitions that we make use of, and explaining the challenges in the online setting. Obtaining a reasonable $\delta$-value does not require updates; for instance, the algorithms in \cite{Bartal96, Bartal2020} guarantee $\delta = O(\log n)$ and can be simulated online. With $O(\log \Delta)$ relevant scales, this gives a distortion of $O(\log n \log \Delta)$.

Offline, one can eliminate the dependence on $\Delta$ by constructing $(O(\log n), O(1/n), 0)$-smooth partitions~\cite{Bartal96,abraham2006advances}. The idea is to first contract all pairs of points at distance $\le \varepsilon s$ into single points\footnote{Any chain of points where each adjacent pair is at distance $\le \varepsilon$ is contracted to a single point.} before partitioning, and expand them back after the partition. This inflates each cluster's diameter by at most $n \varepsilon s$, which contributes an acceptable constant multiplicative factor with $\varepsilon = O(1/n)$.

However, in the online setting, the algorithm may discover only later that two points should have been contracted, yet they were placed in separate clusters at an earlier step. Existing online embeddings thus end up with a distortion either dependent on the aspect ratio or polynomial in $n$. The aim here is to leverage \emph{monotonicity} to overcome these limitations. 

\paragraph{Incremental Setting.}
Monotone embeddings permit us to correct earlier mistakes by merging clusters containing points at a close distance $\le \varepsilon s$ that should have been initially contracted. However, repeatedly merging clusters is problematic, as it can increase cluster diameters and violate the $s$-boundedness property. Therefore, we ultimately might still split some close pairs, contributing to the positive $\gamma$ in the smoothness.

We will use the partitioning algorithm from~\cite{Bartal2020,Bartal96} as the base partition, which we detail in \cref{sec: algorithm outline}, and specify some merging strategy to ensure a small $\gamma$. The algorithm~\cite{Bartal2020,Bartal96} proceeds by constructing balls of random radii around points, which we call \emph{components}. Each point belongs to the first component that includes it. A carefully chosen distribution of the radii ensures $\max \delta^{(j)} = O(\log n)$ and $\sum \delta^{(j)} \le n \max \delta^{(j)}$. Recall from \cref{lemma: hierarchy} that our final distortion is \[
    O\left(\max_{j \in \mathbb{Z}} \delta_t^{(j)} \log \frac{1}{\varepsilon} + \gamma \sum_{j \in \mathbb{Z}} \delta_t^{(j)}\right) = O(\log n \log \varepsilon^{-1} + \gamma n \log n).
    \] Therefore, if we could achieve $\gamma \le O(1/n)$ using $\varepsilon = 1/\poly(n)$, the distortion would be bounded by $O(\log^2 n)$.

Naturally, we attempt to merge two components whenever they contain points forming a pair with distance at most $\varepsilon s$. Crucially, while the multiple merge attempts by a single component are not independent events, the events that \emph{distinct} components simultaneously split close pairs remain independent. Thus, the following surprisingly simple strategy is sufficient, illustrated in \cref{fig: merging}: 

\medskip
\emph{Attempt to merge any pair of components containing points within distance $\varepsilon s$. Allow the first component that splits a close pair to merge freely with any adjacent component, and reject all other merge attempts not involving this component}.
\medskip

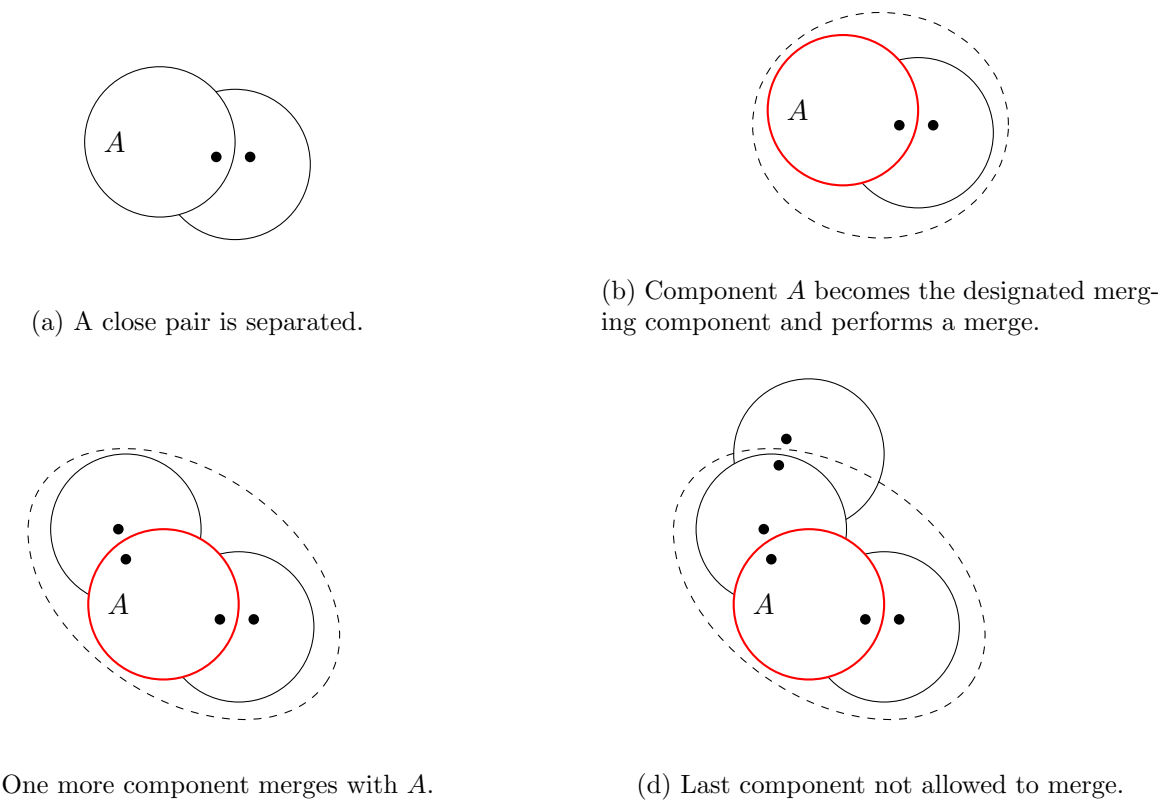
\begin{figure}[t]
  \centering

\begin{subfigure}{0.45\textwidth}
\centering
\begin{tikzpicture}

  \begin{scope}
    \clip (-2,-2) rectangle (3,1.5);
    \clip (-3,-2) rectangle (3,1.5)
          (0,0) circle(1);
    \draw (1,-0.3) circle(1);
  \end{scope}

  \draw (0,0) circle(1);
  \fill (1.2,-0.2) circle(2pt);
  \fill (0.75,-0.2) circle(2pt);

  \node at (-0.6,0) {$A$};
\end{tikzpicture}
\caption{A close pair is separated.}
\end{subfigure}
\hfill
\begin{subfigure}{0.45\textwidth}
\centering
\begin{tikzpicture}

  \draw[dashed] (0.5,-0.2) ellipse (1.7 and 1.5);

  \begin{scope}
    \clip (-2,-2) rectangle (3,1.5);
    \clip (-3,-2) rectangle (3,1.5)
          (0,0) circle(1);
    \draw (1,-0.3) circle(1);
  \end{scope}
  
  \draw[red,thick] (0,0) circle(1);

  \fill (1.2,-0.2) circle(2pt);
  \fill (0.75,-0.2) circle(2pt);
  
  \node at (-0.6,0) {$A$};
\end{tikzpicture}
\caption{Component $A$ becomes the designated merging component and performs a merge.}
\end{subfigure}

\vspace{1em}

\begin{subfigure}{0.45\textwidth}
\centering
\begin{tikzpicture}

  \draw[rotate around={-35:(0.27,0.27)},dashed] (0.27,0.27) ellipse (2.3 and 1.5);

  \begin{scope}
    \clip (-2,-2) rectangle (3,3);
    \clip (-3,-3) rectangle (3,3)
          (0,0) circle(1);
    \draw (1,-0.3) circle(1);
  \end{scope}

  \begin{scope}
    \clip (-2,-2) rectangle (3,3);
    \clip (-3,-3) rectangle (3,3)
          (0,0) circle(1);
    \draw (-0.5,1) circle(1);
  \end{scope}

  \draw[red,thick] (0,0) circle(1);

  \fill (1.2,-0.2) circle(2pt);
  \fill (0.75,-0.2) circle(2pt);
  
  \fill (-0.5,0.6) circle(2pt);  
  \fill (-0.6,1) circle(2pt);  

  \node at (-0.6,0) {$A$};
\end{tikzpicture}
\caption{One more component merges with $A$.}
\end{subfigure}
\hfill
\begin{subfigure}{0.45\textwidth}
\centering
\begin{tikzpicture}

  \draw[rotate around={-35:(0.27,0.27)},dashed] (0.27,0.27) ellipse (2.3 and 1.5);

  \begin{scope}
    \clip (-2,-2) rectangle (3,3);
    \clip (-3,-3) rectangle (3,3)
          (0,0) circle(1);
    \draw (1,-0.3) circle(1);
  \end{scope}

  \begin{scope}
    \clip (-2,-2) rectangle (3,3);
    \clip (-3,-3) rectangle (3,3)
          (0,0) circle(1);
    \draw (-0.5,1) circle(1);
  \end{scope}

  \begin{scope}
    \clip (-2,-2) rectangle (4,3.1);
    \clip (-3,-3) rectangle (4,3.1)
          (-0.5,1) circle(1);
    \draw (0,2) circle(1);
  \end{scope}

  \draw[red,thick] (0,0) circle(1);

  \fill (1.2,-0.2) circle(2pt);
  \fill (0.75,-0.2) circle(2pt);
  
  \fill (-0.5,0.6) circle(2pt);  
  \fill (-0.6,1) circle(2pt);  
  
  \fill (-0.3,2.2) circle(2pt);  
  \fill (-0.4,1.85) circle(2pt);  

  \node at (-0.6,0) {$A$};
\end{tikzpicture}
\caption{Last component not allowed to merge.}
\end{subfigure}
\caption{The merging strategy.}
\label{fig: merging}
\end{figure}

We call the first component that splits a close pair and is allowed to merge the \emph{designated merging component}. Thus, two close points $u, w$ with $d(u, w) \le \varepsilon s$ remain eventually separated only if
\begin{enumerate}[itemsep=0em]
    \item The boundary of a component $A$ splits $u$ and $w$. The probability that a fixed component splits pair $u, w$ is bounded by $O(\log n) \cdot d(u, w)/s$; and
    \item A distinct component $B$ is appointed the designated merging component. A component is the designated merging component only if its boundary splits a close pair $u', w'$ with $d(u', w') \le \varepsilon s$, which for each pair happens with probability $\le O(\varepsilon \log n)$. Thus, the overall probability is at most $O(\varepsilon n^2 \log n) = O(\varepsilon \poly(n))$.
\end{enumerate}
Importantly, these two events are independent after fixing the two components since they only rely on independently sampled radii. A union bound over all pairs then gives $\gamma = O(\poly(n) \varepsilon)$. Taking $\varepsilon = n^{-c}$ for a sufficiently large constant $c$ (e.g., $c = 6$ suffices) achieves that $\gamma = O(1/n)$, as desired.

\medskip
\textit{Handling Unknown \( n \).}  
The embedding algorithm presented above requires prior knowledge of \( n \) to select the parameter \(\varepsilon\). We also provide a modified approach that adapts $\varepsilon$ dynamically to handle unknown \( n \), incurring an additional \(O(\log \log n)\) factor in the contraction. The modified algorithm operates in phases \( i = 1, 2, \dots \), each associated with a guess for the total number of points \( m_i = 2^{2^i} \). Phase \( i \) operates when \(t \in [m_{i-1}+1, m_i]\), hence there are at most $\log \log n$ phases. Instead of permitting only one designated merging component globally, we permit one merging component \emph{per phase}, setting the merging threshold \(\varepsilon_i = 1/\poly(m_i)\) during phase $i$. This modification ensures the partition retains the same smoothness properties but becomes \( O((\log \log n)\cdot s) \)-bounded, thus incurring an extra \( O(\log \log n) \) factor in contraction. A detailed description of this algorithm appears in \cref{appendix: unknown n}.

\paragraph{Fully Dynamic Setting.}
In the fully dynamic scenario (where points arrive and depart), the monotone updates play a distinct yet crucial role: they allow us to bound the number of relevant scales in terms of the width $l$. In this setting, we do not make use of a small $\gamma$ (our $\gamma = 1$). Instead, we only want to bound $\sum \delta^{(j)}$ by controlling the number of relevant scales in which our partitions are not $0$-smooth. For strict embeddings, the number of such scales could be as large as $\Omega(n)$: previously arrived (but now departed) points may have caused undesirable separations. 

We address this problem via merges. The single-merge strategy used in the incremental setting fails here because we no longer have a sufficiently small probability bound (as a function of $l$) on any particular merge attempt. Thus, we employ the following simple merging strategy:

\medskip
\emph{Attempt to merge any pair of components containing points within distance $s/4$. Allow a merge attempt if the merge cannot potentially create a cluster of diameter larger than $s$.}
\medskip

Consequently, when alive points in $L_t$ are either within $s/4$ or further than $s$ apart, close pairs attempt merges, and these attempts succeed because distant points cannot interfere with such merges. This approach ensures that at most $O(l)$ scales with non-zero smoothness $\delta$ exist at any time, resulting in an overall distortion of $O(l \log l)$.

\medskip The overview above describes the algorithms for embedding general metrics. Embedding normed spaces follows an analogous framework with a different base partition that achieves $\delta = O(1)$, detailed in \cref{appendix: normed}.

\subsection{Algorithm Outline}
\label{sec: algorithm outline}
For each scale $s$, we maintain an \emph{online monotone partition} $P_t$ of the set $L_t$ of alive points at time $t$. We construct $P_t$ in two steps: first, we maintain a set $C_t$ of \emph{components}, and second, we maintain a partition $G_t$ over the set $C_t$. Each point $u \in L_t$ is assigned to a component $C_t(u) \in C_t$, and the final partition $P_t$ is given by: $P_t(u) = P_t(w)$ if and only if $G_t(C_t(u)) = G_t(C_t(w))$.

For every component $c$, we ensure that the diameter of the set of points (potentially not alive) that \emph{may} belong to $c$ is bounded by $s/4$. It is helpful to think of $C_t$ as a collection of (incomplete) balls and $C_t(u)$ as the ball containing $u$.

\medskip
To establish $(\delta,\varepsilon,\gamma)$-smoothness, we show:
\begin{enumerate}[itemsep=0em]
    \item If $d(u, w) \le s$, then $\P(C_t(u) \neq C_t(w)) \leq \frac{d(u, w)}{s}\delta $. 
    \item If $d(u, w) \le \varepsilon s$, then $\P(G_t(C_t(u)) \ne G_t(C_t(w))) \le\frac{d(u, w)}{s}  \delta\gamma$.
\end{enumerate}

Intuitively, the set $C_t$ forms our primary partition, and $G_t$ tracks merges among components. We maintain $C_t$ so that $\delta = O(\log n)$. Initially, every component forms a singleton set in $G_t$, and merges are executed when necessary. Formally, merging two components $C_t(u)$ and $C_t(w)$ means replacing $G_t(C_t(u))$ and $G_t(C_t(w))$ in $G_t$ with their union $G_t(C_t(u)) \cup G_t(C_t(w))$. Our algorithms for the incremental and dynamic cases will share the following construction of $C_t$, but have different merging strategies.

\paragraph{Smooth Probabilistic Partition $\boldsymbol{C_t}$.} 
We construct our components $C_t$ by adapting the probabilistic partitions of~\cite{Bartal96, Bartal2020} to accommodate deletions.

Each component in $C_t$ is represented by a triple $(c, r, b)$, where $c \in X$ is the center, $r \in [0, s/8]$ is the radius, and $b \in \mathbb N$ is the birth time of the component. We maintain pairwise distinct birth times for components, ensuring uniqueness. A point $u \in X$ \emph{belongs} to the component identified by $(c, r, b)$ if $d(c,u) \le r$ and $b$ is minimal among all components in $C_t$ containing $u$ within their radius. Note that it is possible for the center not to belong to the component. We sometimes refer to a component simply by its center: for $c \in X$, component $c$ refers to the component with center $c$, and we denote its radius by $r(c)$.

Initially, $C_0 = \varnothing$. When a new point $v \in X$ arrives at time $t$, we first set $C_t = C_{t-1}$ and then add a new component $(v, z, t)$ to $C_t$. The radius $z \in [s/16, s/8]$ is independently sampled from the probability distribution $p(z)$ defined as follows. Letting $j = |C_{t-1}|+1$ denote the number of components after insertion and setting $\chi_j = 2j$, we have:
\[
p(z) = \frac{32\,\chi_j^{2}\log\chi_j}{s(1-\chi_j^{-2})} \exp\left(-\frac{32 z \log \chi_j}{s}\right), \qquad z \in \left[\frac{s}{16}, \frac{s}{8}\right].
\]

When a point departs, we remove from $C_t$ any component that contains no alive points.

The following claim establishes that the smoothness parameter satisfies $\delta_t = O(\log l)$. Its proof follows~\cite{Bartal96, Bartal2020}, incorporating minor adjustments to account for deletions, and is hence deferred to \cref{section: appendix general dynamic}. The technical choice $\chi_j = 2j$ ensures that $\sum_j \chi_j^{-2} \le 1$ and $\log(\chi_j) = O(\log l)$, resulting in the claimed smoothness.

\begin{claim}
\label{clm: bartal padding}
If $d(u,w) \le s$, then at any time $t$, $\P(C_t(u) \neq C_t(w)) \leq O(\log l) \cdot \frac{d(u, w)}{s}.$
\end{claim}

\section{Incremental Setting}
\label{sec: online}
We now formalise the description in \cref{sec: overview} and prove an embedding with $O(\log^2 n)$ distortion in the incremental setting, assuming the knowledge of $n$. 

\begin{theorem}
\label{thm: incremental partition}
For every $n \in \mathbb{N}$, $s > 0$, and $\varepsilon \le 1$, given knowledge of $n$, there is a probabilistic online $s$-bounded monotone partition of up to $n$ points from any metric that is $(O(\log n), \varepsilon, O(n^5\varepsilon))$-smooth at every time $t$. Moreover, if no pair of points in $V_n$ has distance in the interval $[s/16, s]$, the partition is $0$-smooth at all times.
\end{theorem}

\begin{proof}
Our initial component partition $C_t$ is constructed as described in \cref{sec: algorithm outline}, yielding $\delta = O(\log n)$ by \cref{clm: bartal padding}. We first prove the $0$-smooth property.

\begin{claim}
If no pair of points in $V_n$ has distance in $[s/16, s]$, the partition $C_t$ is $0$-smooth at all times $t$.
\label{clm: 0 smooth inc}
\end{claim}
\begin{proof}
In this scenario, points form natural clusters, each with diameter at most \( s/16 \), and distinct clusters are separated by distances exceeding \( s \). Hence, the first component created in each group deterministically includes all points in the group, ensuring that the partition is always \( 0 \)-smooth.
\end{proof}

We now specify the merging strategy. We permit at most one component—termed the \emph{designated merging component}—to merge freely with other components whenever they split a pair of points at a distance at most $\varepsilon s$. All other merge attempts are rejected.

Formally, we maintain (the center of) the designated component \( c^* \). Initially set to \( c^* = \textsc{nil} \). We say a component with center $c$ \emph{cuts} a pair $\{u, w\}$ if exactly one of $u, w$ lies inside the (complete) ball defined by its radius:
\[
r(c) \in [d(c, u), d(c, w)) \cup [d(c, w), d(c, u)).
\]
Note that cut is defined purely in terms of the geometry, and some component cutting a pair is a necessary but insufficient condition for them to belong to different components, since a component created earlier might include both points.

Whenever two points \( u, w \in V_t\) at distance \( d(u, w)\le \varepsilon s \) satisfy \( C_t(u) \ne C_t(w) \), let $c_u$ and $c_w$ denote the centers of $C_t(u)$ and $C_t(w)$, respectively. We proceed as follows:
\begin{enumerate}[itemsep=0em]
    \item If \( c^* = \textsc{nil} \), we assign the designated merging component to be the component among $C_t(u)$ and $C_t(w)$ that cuts $\{u, w\}$. If both components cut $\{u,w\}$, we choose arbitrarily. We then merge \( C_t(u) \) and \( C_t(w) \) by inserting $G_t(C_t(u)) \cup G_t(C_t(w))$ into $G_t$ and removing $G_t(C_t(u))$ and $G_t(C_t(w))$.

    \item If \( c^* \) is already set, we merge \( C_t(u) \) and \( C_t(w) \) only if \( c^* \in \{c_u, c_w\} \); otherwise, we reject the merge and do nothing.
\end{enumerate}

Since each component has a diameter at most $s/4$ and any merged component contains a point that is at most $\varepsilon s$ from the designated merging component, the final clusters will have a diameter at most $3s/4 + 2\varepsilon s \le s$, satisfying $s$-boundedness.

\begin{claim}
\label{clm: enforce}
For any time \( t \) and points \( u, w \in V_t \), if \( d(u, w) \le \varepsilon s \), then
\[
\P(G_t(C_t(u)) \ne G_t(C_t(w))) \le O(n^5 \varepsilon \log n)\cdot\frac{d(u, w)}{s}.
\]
\end{claim}

\begin{proof}
We bound the probability that simultaneously \( C_t(u) \ne C_t(w) \) and the designated merging component \( c^* \) is assigned elsewhere. For each \( c \in V_t \), define the radius set potentially causing \( c \) to cut a close pair and be the designated merging component:
\[
X(c) = \bigcup_{\substack{u', w' \in V_t\\ d(u', w') \le \varepsilon s}} [d(c,u'), d(c,w')).
\]

By a direct integration of the distribution $p$ (cf.~\cref{eq: single split}), we obtain for all $u', w' \in V_t$:
\begin{align*}
\P(c \text{ cuts } \{u', w'\}) &\le \frac{O(\log n)\cdot d(u', w')}{s},\\[6pt]
\P(r(c) \in X(c)) & \le  O(\varepsilon n^2 \log n).
\end{align*}

The event $G_t(C_t(u)) \ne G_t(C_t(w))$ can occur only if (i) some component \( c \) cuts $\{u,w\}$, and simultaneously, (ii) another distinct component \( c' \) is appointed to be the designated merging component $c^*$, which happens only when $r(c') \in X(c')$. Since radii are independently sampled, after fixing $c$ and $c'$, these events are independent. Thus, by union bound:
\belowdisplayskip=-12pt
\begin{align*}
\P(G_t(C_t(u)) \ne G_t(C_t(w)))
&\le \sum_{c \in V_t} \sum_{c' \ne c} \P(c\text{ cuts }\{u, w\}) \P(r(c') \in X(c')) \\
&\le n^2 \cdot \frac{O(\log n)\cdot d(u, w)}{s} \cdot O(\varepsilon n^2 \log n) \\&= O(n^5\varepsilon \log n)\cdot\frac{d(u, w)}{s}. 
\end{align*}
\end{proof}

\cref{thm: incremental partition} then follows from \cref{clm: bartal padding}, \ref{clm: 0 smooth inc}, and \ref{clm: enforce}.
\end{proof}

Applying \cref{lemma: hierarchy} with \( \varepsilon = n^{-6} \) a small enough power of $n$, we achieve the distortion $O(\log ^2 n)$. For embeddings from normed spaces, see \cref{appendix: normed}.

\begin{proof}[Proof of \cref{thm: general} (general metric)]
For each integer \( j \), we apply \cref{thm: incremental partition} with scale \( s = 2^j \) and parameter \( \varepsilon = n^{-6} \). This produces an online \( 2^j \)-bounded partition that is \( (\delta_t^{(j)}, n^{-6}, O(n^{-1})) \)-smooth at every time \( t \). By \cref{lemma: relevant scales}, there are at most \( O(n) \) scales \( j \) for which there exist points \( u,v \in V_n \) satisfying \( d(u,v)\in[2^j/16,2^j] \). \cref{thm: incremental partition} guarantees \( \delta_t^{(j)} = O(\log n) \) for these scales and \( \delta_t^{(j)} = 0 \) for all other scales. Hence, we have
\[
\max_j \delta_t^{(j)} = O(\log n),\qquad \sum_j \delta_t^{(j)} = O(n\log n).
\]

Applying \cref{lemma: hierarchy}, the total distortion is bounded by
\[
O(\log n \log(n^6)) + O(n\log n \cdot n^{-1}) = O(\log^2 n). \qedhere
\]
\end{proof}

\section{Fully Dynamic Setting}
\label{sec: dynamic}

When the width $l$ of the update sequence is bounded, it is natural to expect distortion to depend only on $l$, independently of the sequence length $n$. However, this is impossible with traditional strict embeddings: there exist update sequences with width $l=3$ that incur deterministic lower bounds of $\Omega(2^n)$ and randomized lower bounds of $\Omega(n)$. In fact, the distortion remains unbounded even for deterministic monotone embeddings.

\DynamicLowerDeter*
\begin{proof}
The proofs for strict embeddings appear in \cref{thm: online_lower} and \cref{thm: online_lower_random}. Intuitively, for strict embeddings, points departing won't help the algorithm since the adversary could just keep the most distorted pairs.

For deterministic monotone embeddings, we give an explicit update sequence with $n$ points and width $3$ that enforces $\Omega(n)$ distortion. The points lie on the real line: At times $1, 2, 3$, points arrive at coordinates $0, 1, 2$. Subsequently, whenever the current points have coordinates $(0, x, x+1)$, the point at position $x$ departs and a new point at $x+2$ arrives. This step repeats until a point appears at coordinate $n$. Intuitively, this configuration maintains one fixed point at the origin, while two others iteratively shift rightward by increments of $1$.

Without loss of generality, deterministic embeddings can be assumed non-contractive, as otherwise we could scale all distances by the maximum contraction. Let $t_x$ denote the arrival time of the point at coordinate $x$. It is easy to verify that $d_{t_n}(0, n) \le \max_{i=0}^{n-1} {d_{t_{i+1}}}(i, i+1)$ inductively due to HST being an ultrametric and monotonicity only allowing the distance to decrease. Since $d_{t_n}(0, n) \ge n$, there exists some distances expanded by $\Omega(n)$.
\end{proof}

Despite these negative results, we show that combined with randomness, monotone embeddings achieve distortion bounded solely by the width $l$.

\subsection{Upper Bound}
\label{sec: general dynamic}
\begin{theorem}
\label{thm: general partition dynamic}
For every $l \in \mathbb{N}$ and $s > 0$, there is a probabilistic online $s$-bounded monotone partition for update sequences of width $l$ that is $O(\log l)$-smooth at every time $t$. Moreover, if no pair of points $u,w \in L_t$ satisfies $d(u,w)\in[s/4,s]$, the partition $P_t$ at time $t$ is $0$-smooth. These results hold without prior knowledge of $l$.
\end{theorem}

\begin{proof}
We follow the framework outlined in \cref{sec: algorithm outline}, constructing $C_t$ as previously described, and now specifying the merge procedure for constructing $G_t$ explicitly.

In contrast to the incremental case, where a single merging component sufficed, we allow multiple merges to maintain a bounded number of relevant scales. At each time $t$, initialize $G_t$ as a copy of $G_{t-1}$. When a new point $v_t$ arrives, we first update $C_t$ according to \cref{sec: algorithm outline}. If there exists some point $u \in L_t$ with $d(u,v_t)\le s/4$ but $C_t(u)\neq C_t(v_t)$, we attempt to merge the components containing $u$ and $v_t$. The merge succeeds if and only if the merged cluster's diameter (accounting for points potentially joining in the future due to their pre-selected radii) remains at most $s$. Formally, the condition for a successful merge is that for every $(c_1, r_1, b_1) \in G_t(C_t(u))$ and $(c_2, r_2, b_2) \in G_t(C_t(v_t))$,
\[
r_1 + r_2 + d(c_1, c_2) \le s.
\]
This condition ensures the $s$-boundedness of the partition persists despite potential future arrivals. When a point departs, we remove empty components from $C_t$, potentially allowing previously rejected merges to now succeed; we thus re-check merge conditions after each departure.

\begin{claim}
\label{lemma: zero smooth}
If no pair of alive points in $L_t$ has distance in $[s/4, s]$, then $P_t$ is $0$-smooth at time $t$.
\end{claim}
\begin{proof}
Suppose no alive points $u,w \in L_t$ have distance in $[s/4,s]$. Points in $L_t$ naturally form clusters of diameter at most $s/4$, with distances between clusters strictly exceeding $s$. For any pair $u',w'$ within distance at most $s/4$ belonging initially to different components $C_t(u')$ and $C_t(w')$, the merging condition is satisfied since no distant clusters can obstruct the merge. Thus, all pairs of points within each cluster eventually merge, ensuring $P_t$ is $0$-smooth at time $t$.
\end{proof}

This concludes the proof of \cref{thm: general partition dynamic}.
\end{proof}

Using \cref{lemma: relevant scales}, at every time $t$, we have that $\sum \delta_t^{(j)} = O(l \log l)$. Applying \cref{lemma: hierarchy} gives the general metric part in \cref{thm: general dynamic}. Note we never utilize prior knowledge of $l$, hence neither does \cref{thm: general dynamic}. 

\subsection{Lower Bound}

It might seem tempting to apply similar techniques as in the incremental setting to achieve a distortion of $O(\polylog l)$. Unfortunately, this is not possible. To illustrate why, consider the following scenario (see Figure~\ref{fig: bad case}): Suppose two points $v_1$ and $v_2$ on the line lie in separate components, and two additional close points move through from left to right. At some stage, these two moving points will cause a merge of the intervals containing $v_1$ and $v_2$. Due to the monotonicity constraint, this merge cannot be undone, even after the moving points leave, thus preventing future merges that may be necessary when other close points arrive. 

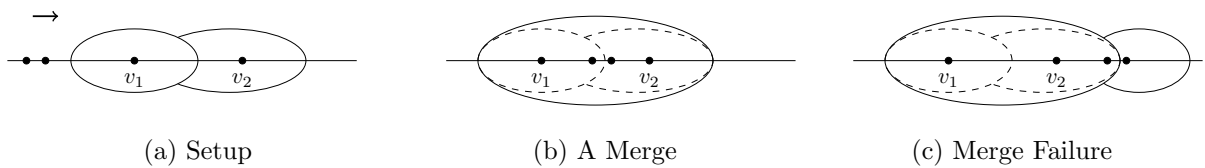
\begin{figure}[htb]
\centering
    \begin{subfigure}{0.32\linewidth}
    \centering
    \noindent\resizebox{\textwidth}{!}{
    \begin{tikzpicture}[font=\small]
        \path[use as bounding box] (-1,-0.95) rectangle (5,1.05);
        \draw (-1,0) -- (4.5,0);

        \draw[->,thick] (-0.6, 0.7) -- (-0.2, 0.7);

        \filldraw (-0.4,0) circle (1.5pt);
        \filldraw (-0.7,0) circle (1.5pt);
        \filldraw (1,0) circle (1.5pt) node[below=2pt]{$v_1$};
        \filldraw (2.7,0) circle (1.5pt) node[below=2pt]{$v_2$};

        \begin{scope}
            \clip (-1,-1) rectangle (4.5,1); 
            \clip
                (-1,-0.6) rectangle (4.5,0.6) 
                (1,0) ellipse (1 and 0.5); 
            \draw (2.5,0) ellipse (1.2 and 0.5);
        \end{scope}

        \draw (1,0) ellipse (1 and 0.5);
    \end{tikzpicture}
    }
    \caption{Setup}
\end{subfigure}
    \begin{subfigure}{0.32\linewidth}
        \centering
        \noindent\resizebox{\textwidth}{!}{
        \begin{tikzpicture}[font=\small]
        \path[use as bounding box] (-1,-0.95) rectangle (5,1.05);
        \draw (-0.5,0) -- (5,0);


        \filldraw (2.1,0) circle (1.5pt);
        \filldraw (1.8,0) circle (1.5pt);
        \filldraw (1,0) circle (1.5pt) node[below=2pt]{$v_1$};
        \filldraw (2.7,0) circle (1.5pt) node[below=2pt]{$v_2$};

        \begin{scope}
            \clip (-1,-1) rectangle (4.5,1); 
            \clip
                (-1,-0.6) rectangle (4.5,0.6) 
                (1,0) ellipse (1 and 0.5); 
            \draw[dashed] (2.5,0) ellipse (1.2 and 0.5);
        \end{scope}

        \draw[dashed] (1,0) ellipse (1 and 0.5);
        \draw (1.85,0) ellipse (1.85 and 0.7);
    \end{tikzpicture}
    }
        \caption{A Merge}
    \end{subfigure}
    \begin{subfigure}{0.32\linewidth}
        \centering
        \noindent\resizebox{\textwidth}{!}{
        \begin{tikzpicture}[font=\small]
        \path[use as bounding box] (-1,-0.95) rectangle (5,1.05);
        \draw (-0.5,0) -- (5,0);


        \filldraw (3.5,0) circle (1.5pt);
        \filldraw (3.8,0) circle (1.5pt);
        \filldraw (1,0) circle (1.5pt) node[below=2pt]{$v_1$};
        \filldraw (2.7,0) circle (1.5pt) node[below=2pt]{$v_2$};

        \begin{scope}
            \clip (-1,-1) rectangle (4.5,1); 
            \clip
                (-1,-0.6) rectangle (4.5,0.6) 
                (1,0) ellipse (1 and 0.5); 
            \draw[dashed] (2.5,0) ellipse (1.2 and 0.5);
        \end{scope}
        \begin{scope}
            \clip (-1,-1) rectangle (5,1); 
            \clip
                (-1,-0.6) rectangle (5,0.6) 
                (1.85,0) ellipse (1.85 and 0.7); 
            \draw (4,0) ellipse (0.8 and 0.5);
        \end{scope}

        \draw[dashed] (1,0) ellipse (1 and 0.5);
        \draw (1.85,0) ellipse (1.85 and 0.7);
    \end{tikzpicture}
        }
        \caption{Merge Failure}
    \end{subfigure}
    \caption{Illustration of obstacles in the dynamic setting.}
    \label{fig: bad case}
\end{figure}

In fact, we show that a linear dependence on $l$ is unavoidable. Notably, this lower bound holds even in an \emph{offline} setting, as it uses a fixed sequence.

\DynamicLowerRand*

\begin{proof}
We construct the sequence on the real line, identifying each point with its coordinate. Let $m = 2^l$. Define the following notion: for each integer $x$, the set of \emph{encompassing points} is given by
\[
E(x) = \bigcup_{j=0}^l \left\{\left\lfloor\frac{x}{2^j}\right\rfloor  2^j,\; \left(\left\lfloor\frac{x}{2^j}\right\rfloor+1\right) 2^j\right\}.
\]
Intuitively, $E(x)$ includes endpoints obtained by recursively halving the segment containing $x$ exactly $l$ times (see \cref{fig: encompassing}).

\begin{figure}[ht]
        \centering
        \resizebox{\linewidth}{!}{
        \begin{tikzpicture}[font=\small]
            \draw (-0.5,0) -- (16.5,0) ;
        
            \foreach \x/\y in {0/0,4/8,5/10,6/12,8/16,16/32} {
               \draw (\x,0.2) -- (\x,-0.2) node[below=2pt]{\y};
            }
            \filldraw (5.5,0) circle (2.5pt);
            
            
        \end{tikzpicture}
        }
        \caption{Encompassing points marked as vertical segments.}
        \label{fig: encompassing}
    \end{figure}
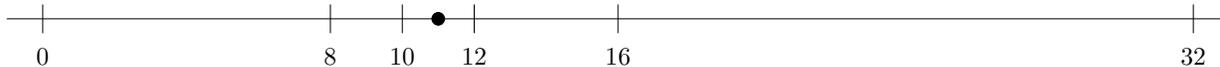

Our update sequence models the following process, depicted in \cref{fig: dynamic random lower bound example}: two moving points travel from $0$ to $m$, arriving and leaving iteratively: initially, points $0$ and $1$ arrive, then point $0$ leaves and point $2$ arrives, and so forth. At each step, when the moving points are at coordinates $(x, x+1)$, we ensure the points in $E(x)$ are present and remove any points no longer belonging to $E(x)$ (see \cref{fig: dynamic random lower bound example}). Clearly, at any given time, the number of alive points is $O(l)$.

\begin{figure}[htb]
\centering
    \begin{subfigure}{\linewidth}
        \vspace{15pt}
        \centering
        \noindent\resizebox{\textwidth}{!}{
        \begin{tikzpicture}[font=\small]
            \draw (-1,0) -- (16.5,0) ;
        
            \foreach \x in {0,4,5,6,8,16} {
                \draw[fill=white] (\x,0) circle (2.5pt);
            }
            \filldraw (4.4,0) circle (2.5pt);
            \filldraw (4.7,0) circle (2.5pt);
            \draw [thick,->] (4.6, 0.6) -- (5.2, 0.6);
        \end{tikzpicture}
        }
    \end{subfigure}
    \begin{subfigure}{\linewidth}
        \vspace{15pt}
        \centering
        \noindent\resizebox{\textwidth}{!}{
        \begin{tikzpicture}[font=\small]
            \draw (-1,0) -- (16.5,0) ;
        
            \foreach \x in {0,8,12,14,15,16} {
                \draw[fill=white] (\x,0) circle (2.5pt);
            }
            \filldraw (14.3,0) circle (2.5pt);
            \filldraw (14.6,0) circle (2.5pt);
            \draw [thick,->] (14.4, 0.6) -- (15, 0.6);
        \end{tikzpicture}
        }
    \end{subfigure}
    
    \caption{Black points are iterating to the right. White points are encompassing points.}
    \label{fig: dynamic random lower bound example}
\end{figure}
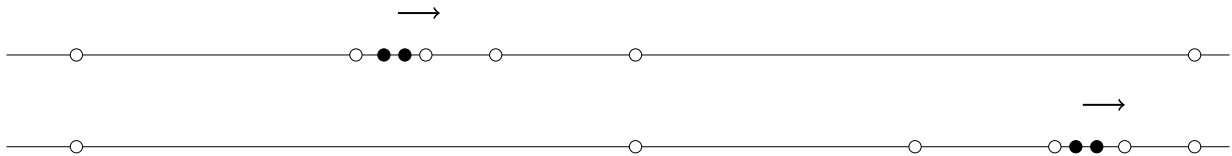

Since we use a fixed sequence, assuming the embedding is non-contractive is without loss of generality. By Yao's minimax principle, it suffices to present a distribution $\D$ over tuples $(u,v,t)$ such that, for any deterministic embedding $\{d_t\}$, we have $\E_{(u,v,t)\sim \D}[d_t(u,v)/d(u,v)] \ge \Omega(l)$. Let $t_x$ be the first time point $x$ appears as an iterating point (it may appear earlier as an encompassing point). Define the distribution $\D$ to be uniform over triples $(x, x+1, t_{x+1})$. Since $d(u, v) = 1$ for all $(u, v, t)$ in the support of $\D$, we can drop it and bound $f(l)=\E_{(u,v,t)\sim \D}[d_t(u,v)]$ the distortion of the sequence corresponding to $l$. We will show inductively that $f(l) \ge l/2$. The base case $l=2$ is trivial. 

We now perform the induction as follows. Fix $l$ and $m=2^l$ and assume $f(j) \ge j/2$ for all $j < l$. Since the embedding is non-contractive, we have $d_{t_m}(0,m)\ge m$. Let $p$ be the largest integer $p < m$ such that $d_{t_p}(0,p)<m$; thus, we must have $d_{t_{p+1}}(0,p+1) \ge m$. As an HST is an ultrametric and $d_{t_{p+1}}(0, p) \le d_{t_p}(0, p) < m$, it follows directly that $d_{t_{p+1}}(p,p+1)\ge m$. The pair $(p,p+1)$ yields a distortion of $m$ and is selected with probability $1/m$, contributing at least $1$ to the expected distortion.

Now, we consider the intervals $[1, p]$ and $[p+1, m]$. Note that we could not invoke the inductive hypothesis immediately, because our instance requires endpoints to exist when the points are iterating. Nevertheless, we can partition the intervals into smaller intervals of length powers of two, in which the encompassing points will be persistent endpoints.

\begin{figure}[ht]
        \centering
        \resizebox{\linewidth}{!}{
        \begin{tikzpicture}[font=\small]
            \draw (-0.5,0) -- (16.5,0) ;
        
           \draw (0,0.2) -- (0,-0.2) node[below=2pt]{0};
           \draw (16,0.2) -- (16,-0.2) node[below=2pt]{256};
            \filldraw (5,0) circle (1pt) node[below=2pt]{};
            \filldraw (5.2,0) circle (1pt) node[below=2pt]{};
            \draw [thick, decorate,decoration={brace,amplitude=10pt},yshift=2pt] (9,0) -- (16,0) node [black,midway,yshift=16pt] {$128$};
            \draw [thick, decorate,decoration={brace,amplitude=10pt},yshift=2pt] (6.6,0) -- (9,0) node [black,midway,yshift=16pt] {$32$};
            \draw [thick, decorate,decoration={brace,amplitude=10pt},yshift=2pt] (5.4,0) -- (6.6,0) node [black,midway,yshift=16pt] {$16$};
            \draw [thick, decorate,decoration={brace,amplitude=10pt},yshift=2pt] (5.2,0) -- (5.4,0) node [black,midway,yshift=16pt] {$1$};
            \draw [thick, decorate,decoration={brace,amplitude=10pt},yshift=2pt] (0,0) -- (3.8,0) node [black,midway,yshift=16pt] {$64$};
            \draw [thick, decorate,decoration={brace,amplitude=10pt},yshift=2pt] (3.8,0) -- (4.4,0) node [black,midway,yshift=16pt] {$8$};
            \draw [thick, decorate,decoration={brace,amplitude=10pt},yshift=2pt] (4.4,0) -- (4.775,0) node [black,midway,yshift=16pt] {$4$};
            \draw [thick, decorate,decoration={brace,amplitude=10pt},yshift=2pt] (4.775,0) -- (5,0) node [black,midway,yshift=16pt] {$2$};
            
        \end{tikzpicture}
        }
        \caption{Example of decomposition with $l=8$, $p=78$. Not drawn to scale.}
        \label{fig: example}
    \end{figure}
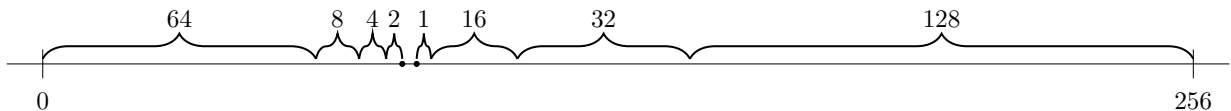

Let $p$ have binary representation $p=\sum_{j=1}^h 2^{w_j}$ for distinct integers $w_1>w_2>\dots>w_h$ and $p_i=\sum_{j=1}^i 2^{w_j}$ for $i \in [0, h]$. Then the pairs $(p_i, p_{i+1})$ will persist as encompassing pairs when $t \in [t_{p_i}, t_{p_{i+1}}]$, and the instance forms a smaller replica of size $w_{i+1}$. Similarly, we decompose the interval $[p+1, m]$. See \cref{fig: example} for an example of the decomposition.

We can now safely invoke our inductive hypothesis. For any $1\le p\le m-1$ being the split point, the binary representations of $p$ and $m-p-1$ contain each $2^0, \ldots, 2^{l-1}$ exactly once, since the binary representation of $2^l-1$ is $(111\ldots1)_2$ and thus $p$ and $2^{l}-1-p$ would not have $1$ on the same bit. For the interval of size $2^j$, the tuple $(u, v, t) \sim \D$ falls in that interval with probability $2^{j-l}$; conditioned on the pair belonging to the interval, the expected distortion is $f(j)$. Recall we also have a probability of $2^{-l}$ to pick the pair $(p, p+1)$ whose distance is $2^l$. We have 
\[
   f(l) = \sum_{j=0}^{l-1} 2^{j-l} \cdot f(j) + 2^{-l}\cdot 2^l\ge\sum_{j=0}^{l-1}\frac{j}{2^{l-j+1}}+1=\frac{l}{2}-1+\frac{1}{2^l}+1\ge\frac{l}{2}. \qedhere
\]

\end{proof}

\section{Applications}
\label{sec: application}

\subsection{Metrical Request-Answer Games}
Throughout this section, we use bold notation such as $\seq a$ for finite sequences. For a sequence $\seq{a}$ of length $m$ and any $t \le m$, we denote by $a_t$ the $t$th entry and by $\seq a_{[0, t]}$ the prefix of length $t$ of $\seq a$.

We define a general class of online problems in metric spaces that we call \emph{metrical request-answer games}. It generalizes the notion of request-answer games defined in~\cite{BenDavidBKTW94} by introducing dependence on a metric space. The request-answer games in~\cite{BenDavidBKTW94} correspond to the special case of our definition where $\alpha_t=0$.

\begin{definition}[Metrical Request-Answer Game]\label{def:requestAnswer}
    A \emph{metrical request-answer game} is defined by a request set $R$, an answer set $A$, a metric space $(X, d_X)$, and for each $t\in\mathbb N$ a cost function $c_t: R^t \times A^t \to \R_{\ge 0}\cup \br\infty$ of the form\footnote{In fact, as long as $c_t$ is non-negative, concave, and non-decreasing in $d_X$, our theorem holds.}
        \[
        c_t(\seq{r}, \seq{a}) 
        = \sum_{(u,v)\in X^2} \alpha_t(u, v, \seq{r}, \seq{a})d_X(u, v) 
          + \beta_t(\seq{r}, \seq{a}),
        \]
        for some functions $\alpha_t\colon X^2\times R^t\times A^t\to \R_{\ge 0}$ and $\beta_t\colon R^t\times A^t\to \R_{\ge 0}\cup \br\infty$.
\end{definition}

Intuitively, $\alpha_t(u, v, \seq r, \seq a)$ indicates the number of times an algorithm pays the distance from $u$ to $v$ at step $t$ if it serves the request sequence $\seq r$ with answers $\seq a$, and $\beta_t(\seq r, \seq a)$ is some metric-independent cost. In particular, $\beta_t(\seq r,\seq a)=\infty$ allows to specify infeasible answers.

\begin{definition}[Online Algorithm]
    A \emph{deterministic online algorithm} $\Alg$ for a metrical request-answer game is a sequence of functions $\Alg_t\colon R^t \to A$. Given a request sequence $\seq{r}\in R^m$, we write $\Alg(\seq{r}) = (a_1, a_2,\dots,a_m)$ for the sequence of answers selected by $\Alg$,
    where $a_t = \Alg_t(\seq r_{[0, t]})$ is the answer after the $t$-th request. The cost of $\Alg$ on $\seq{r}$ is $\mathrm{cost}_{\Alg}(\seq{r}) 
    = \sum_{t=1}^m c_t(\seq r_{[0, t]}, \Alg(\seq r_{[0, t]}))$.
    The optimal cost for the same sequence is
    $\mathrm{opt}(\seq{r}) 
    = \min_{\seq{a}\in A^m}\sum_{t=1}^m c_t(\seq r_{[0, t]}, \seq a_{[0, t]})$.

    A \emph{randomized online algorithm} $\Alg$ is a distribution over deterministic online algorithms $\Alg_x$. For any request sequence $\seq{r}$, the answer sequence $\Alg(\seq{r})$ and hence $\mathrm{cost}_{\Alg}(\seq{r})$ become random variables. Algorithm $\Alg$ is \emph{$\rho$-competitive} if there exists $\eta\ge 0$ such that for every $\seq{r}$,
    $
    \E_x[\mathrm{cost}_{\Alg_x}(\seq{r})] 
    \le \rho\cdot \mathrm{opt}(\seq{r}) + \eta
    $. We sometimes use $\mathrm {cost}_{\Alg}(\seq r)$ to refer to $\E_x[\mathrm{cost}_{{\Alg}_x}(\seq{r})]$.
\end{definition}

Countless online problems involving metric spaces can be modelled as metrical request-answer games. Some examples are given below.

\begin{example}[Metrical Task Systems (MTS) \cite{BorodinLS92}]
    A metrical task system (MTS) is defined on a metric space $(X,d)$, whose points are called \emph{states}. There is a fixed initial state $a_0\in X$. At each time step $t$, a task arrives, specified by a function $r_t: X \to \R_{\geq 0}\cup\{\infty\}$ that describes the cost of processing task $t$ in each state. In response, an algorithm has to choose a state $a_t\in X$ in which to process the task, incurring cost $d(a_{t-1},a_t)$ for movement and $r_t(a_t)$ for processing the task. 
    
    To formulate MTS as a metrical request-answer game, the request set $R$ is the set of all functions $X\to \R_{\geq 0}\cup\{\infty\}$, the answer set is $X$, and we define $\alpha_t(u, v, \seq r, \seq a) = 1$ iff $(a_{t-1}, a_t) = (u, v)$ and $0$ otherwise, and $\beta_t(\seq r, \seq a) = r_t(a_t)$.
\end{example}

We note that MTS is already a very general framework, containing many important online problems as special cases (e.g., $k$-server, caching, convex function chasing, layered graph traversal, dynamic power management, etc \cite{Vaze_2023, antoniadis2016chasing}). Consequently, all of these problems can also be modelled as metrical request-answer games.

\begin{example}[$k$-Server \cite{MMS88} and $k$-Taxi \cite{fiat1990}]
    In the $k$-taxi problem, there are $k$ taxis located at points of a metric space $(X, d_X)$. At each time $t$, a request appears, specified by a pair $(x_t,y_t)\in X\times X$ representing a passenger that wants to travel from $x_t$ to $y_t$. In response, an algorithm must move a taxi to $x_t$ and then to $y_t$ before seeing future requests. The cost is the total distance traveled \emph{without} a passenger on board (i.e., the distance from $x_t$ to $y_t$ is excluded). The $k$-server problem is the special case where $x_t=y_t$ for each request.
    
    To model $k$-taxi as a metrical request-answer game, we choose $R=X\times X$, $A=\{1,\dots,k\}$ (assigning numbers to taxis in some fixed way), $\beta_t(\seq{r},\seq{a}) = 0$, and $\alpha_t(u,v,\seq{r},\seq{a}) = 0$ or $1$ depending on if a taxi moves from $u$ to $v$ without a passenger in the corresponding step.
\end{example}

\begin{example}[Abstract Network Design]
     Metrical request-answer games subsume the class of \emph{Abstract Network Design} problems introduced in \cite{Bartal2020} in the context of metric embeddings. We omit a formal definition here. It includes problems like $k$-server, Steiner forest and variants, reordering buffer management, constrained file migration, and others.
\end{example}

\begin{example}[Metric Problems with Delay]
    Problems whose costs combine distance and delay costs, such as online service with delay~\cite{AzarGGP21} and online matching with delay~\cite{EmekKW16}, can be modeled as metrical request-answer games, charging delay costs via the $\beta_t$ functions. With each answer, the algorithm specifies when and how it would serve each pending request assuming no additional information (in the form of new requests) becomes available beforehand.
\end{example}

For a request sequence $\seq r \in R^t$, let $V(\seq{r})$ denote the set of relevant points induced by the first $t$ requests, i.e., points whose distance to some point might contribute to the cost:
\[
V(\seq{r}) 
= \left\{
  u\in X \mid
  \exists j\le t, v\in X,\ \seq{a}\in A^j:
  \alpha_j(u,v,\seq r_{[0, j]},\seq a) + \alpha_j(v,u,\seq r_{[0, j]},\seq a) > 0
\right\}.
\]
For example, in the $k$-taxi problem, $V(\seq r)$ is the set of points appearing in the initial configuration and requests of $\seq r$.

\subsection{Monotone Potential Functions}
\label{sec: monotone potential}
We have alluded to the fact that potential functions are usually monotone in distances as a central motivation for online monotone embeddings. In this section, we provide more discussion.

We first define the notion of a potential function, which generalizes the definition in~\cite{BenDavidBKTW94} to include online algorithms against oblivious adversaries. Recall that $\mathcal D(S)$ denotes the set of probability distributions over a set $S$.

\begin{definition}[Potential Function]
For a metrical request-answer game, a family
\[
\Phi = \{\Phi_t : R^t \times \D(A^t) \times A^t \to \R_{\ge 0}\}_{t\ge 0}
\]
is called a potential function for $\rho$-competitiveness if the following is true for each $t\ge 1$:

For every $\seq r\in R^{t}$ and $\gamma\in \D(A^{t-1})$, there exists $\gamma'\in\D(A^{t})$ such that the marginal distribution of $\gamma'$ over the first $t-1$ answers is $\gamma$ and for every $\seq b \in A^{t}$ we have
  \begin{align}
  \E_{\seq a\sim \gamma'}\left[c_{t}(\seq r, \seq a)\right]
  + \Phi_{t}(\seq r, \gamma', \seq b)
  - \Phi_{t-1}(\seq r_{[0, t-1]}, \gamma, \seq b_{[0, t-1]})
  \le \rho \cdot c_{t}(\seq r, \seq b).\label{eq:potential}
  \end{align}
\end{definition}

Potential functions are a common method of proving the competitiveness of online algorithms. Given a potential function, a corresponding online algorithm can be defined by extending its distribution of answers in each step from $\gamma$ to some $\gamma'$ satisfying inequality~\eqref{eq:potential}. Taking $\seq b$ to be the answer sequence chosen by some optimal offline algorithm and summing inequality~\eqref{eq:potential} over all time steps shows that the algorithm is $\rho$-competitive. 

Intuitively, the potential function measures the disadvantage of the online algorithm in a given configuration. Typically, this involves a \emph{distance} between the online and offline configurations (e.g., matching or relative entropy), and sometimes \emph{dispersion} of the online configuration, as this causes future uncertainty or mistakes (e.g., pairwise distances among servers of the Double Coverage algorithm~\cite{CKP91new} for $k$-server), both of which are monotone. We give a few examples to illustrate the idea.

\begin{example}[Potential for $k$-Server on Trees]
\label{ex:DCPotential}
    For the $k$-server problem on trees, \cite{CKP91new} gave a $k$-competitive algorithm, using a potential $\Phi = kM + \Sigma$  where $M$ denotes the value of a minimum matching between the locations of the $k$ online and offline servers and $\Sigma$ denotes the sum of pairwise distances between the online servers.
\end{example}

\begin{example}[Potential for $k$-Server on HSTs]
\label{ex:serverPotential}
    For the $k$-server problem on HSTs, \cite{Bubeck18} gave an $O(\log^2 k)$-competitive randomized algorithm, using a monotone potential. The full description and proof of monotonicity appear in \cref{appendix: kserver}.
\end{example}

\begin{example}[Potential for $k$-Taxi on HSTs]\label{ex:taxiPotential}
    For the $k$-taxi problem on HSTs, \cite{coester2019online} gave a $(2^k-1)$-competitive randomized algorithm, using a potential equal to $2^k-1$ times the value of a minimum matching between the locations of the $k$ online taxis and the locations of the $k$ offline taxis. In our notation, we can express this as follows:  
    
    For two multisets $A$ and $B$ of locations in an HST $T$, let $d_T(A, B)$ denote the value of a minimum matching between $A$ and $B$. For a sequence of requests $\seq r\in R^t$ and answers $\seq a\in A^t$, let $C(\seq r, \seq a)$ denote the resulting set of locations of the $k$ taxis. Then $\Phi_t(\seq r, \gamma, \seq b) = (2^k-1)\E_{\seq a \sim \gamma}[d_T(C(\seq r, \seq a), C(\seq r, \seq b))]$.
\end{example}

\begin{example}[Potential for Constrained Forest on Trees]\label{ex:forestPotential}
In the Constrained Forest Problem~\cite{KP95Constrained, Bartal2020}, each request consists of a set of terminals $Z_t$ and a cut requirement function $g_t: 2^{Z_t} \rightarrow \{0,1\}$ that is proper.\footnote{The function $g_t$ is proper if
$g_t(\varnothing) = g_t(Z_t) = 0$, 
$g_t(X) = g_t(Z_t - X)$ for all $X \subseteq Z_t$, and 
$g_t(X \cup Y) \leq \max\{g_t(X), g_t(Y)\}$ for all disjoint sets $X, Y \subseteq Z_t$.} The problem is augmented with a subadditive cost function $\chi_t: 2^{\{1, \ldots, t\}} \to \R$ that is gradually revealed.

The algorithm responds with a set of edges $R_t$ that satisfies $(Z_t, g_t)$, i.e., for every vertex subset $S \subseteq V$ such that $g_t(S \cap Z_t) = 1$, there is at least one edge in $R_t$ that leaves the set $S$. The cost of the algorithm until time $t$ is $\sum_{u, v} d(u, v) \chi(\{t' \le t \mid (u, v) \in R_t\})$.

    The subadditive constrained forest problem can be trivially solved exactly on trees~\cite{Bartal2020}. Therefore, the constant function $0$ is a potential function.
\end{example}

\subsection{The Application Framework}
\label{sec: reduction}
We are now ready to state our reduction theorem, which is the formal version of Theorem~\ref{thm:informalApplication}: our evolving embeddings are applicable whenever there is a potential function in the family of target spaces that is non-decreasing in distances.

\begin{theorem}
\label{thm: reduction}
Consider a request set $R$, answer set $A$, and functions $\alpha_t$ and $\beta_t$ as in Definition~\ref{def:requestAnswer} for some fixed ground set $X$. All metric spaces considered in this theorem have subsets of $X$ as their sets of points. For a metric space $M=(V_M,d_M)$, denote by $G_M$ the associated request-answer game (restricted to request sequences $\seq r$ with $V(\seq r)\subseteq V_M$) and by $c_t^M$ its cost function (as induced by $\alpha_t$ and $\beta_t$).
Let $\mc R$ be some set of request sequences such that $|V(\seq r)|\le n$ for all $\seq r\in \mc R$. Let $\mc M$ be a family of target metrics. Suppose the following holds:
\begin{enumerate}
\item There is an online monotone embedding from $N=(X, d_N)$ to $\mc M$ with distortion $\lambda$ for sequences of $n$ points.
\item For each $M \in \mc M$,  there exists a potential function $\Phi^M$ for $\rho$-competitiveness for $G_M$. Further, the family $\{\Phi^M\}$ is non-decreasing in distances: if $M$ dominates $M'$ (i.e., $d_M(u,v)\ge d_{M'}(u,v)$ for all $u,v$), then for all $t, \seq r, \gamma, \seq b$,
\[
\Phi^M_t(\seq r,\gamma,\seq b) \ge \Phi^{M'}_t(\seq r,\gamma,\seq b).
\]
\end{enumerate}
Then there is a $\rho\lambda$-competitive algorithm for $G_{N}$ for request sequences in $\mc R$.
\end{theorem}
\begin{proof}[Proof Sketch]
    Let $r\in\mc R$ be the request sequence that is revealed online.
    Note that the sets $V(\seq r_{[0, t]})$ of relevant points for the first $t$ requests are increasing in $t$. Applying the online monotone embedding, we obtain a metric $M_t=(V(\seq r_{[0, t]}),d_t)$ for each $t=1,\dots,m$, where $m$ is the length of $\seq r$.

    For fixed $\seq M=(M_1,\dots,M_m)$, we denote by $G_{\seq M}$ the game with cost function $c_t^{M_t}$ at step $t$. We first define an algorithm $\Alg^{\seq M}$ for $G_{\seq M}$ inductively as follows. If $\gamma\in\mc D(A^{t-1})$ is the distribution of answers before the $i$th request arrives, then the next answer is chosen so as to extend the distribution to $\gamma'\in\mc D(A^{i})$ satisfying inequality~\eqref{eq:potential} for the cost function $c_{t}^{M_{t}}$ and potential $\Phi^{M_t}$. The overall algorithm $\Alg$ for $G_N$ is obtained by taking randomness over $\seq M$.

    Bounding the competitive ratio has three steps. First, we show that the total cost of $\Alg^{\seq M}$ in $G_{\seq M}$ is bounded by $\rho$ times the offline cost in $G_{\seq M}$. This step uses the monotonicity of the potential function and the embedding. Then, we bound the expected offline cost in $G_{\seq M}$ by the offline cost on $G_N$ times the distortion of $\seq M$. Finally, since the embedding is non-contractive, we bound the cost of $\Alg$ for the original game $G_N$ by its cost on $G_{\seq M}$. Combining the three arguments gives the competitive ratio $\rho\lambda$. 
    
    The full proof formalizes the above reasoning and appears in \cref{appendix: application}.
\end{proof}

Together with the monotonicity of \cref{ex:forestPotential}, \cref{thm: reduction} immediately implies a competitive algorithm for the subadditive constrained forest problem, recovering a result from~\cite{Bartal2020,Bhore24duet} where this was proved using a strict online embedding and a min operator to combine with a baseline algorithm to avoid dependence on $\Delta$.

\begin{restatable}[Constrained Forest~\cite{Bartal2020,Bhore24duet}]{corollary}{CForest}
\label{thm: constraint forest}
There is an $O(\log^2 n)$-competitive algorithm for the subadditive constrained forest problem on general metrics and an $O(\log n)$-competitive algorithm on $O(1)$-dimensional normed spaces.
\end{restatable}

\cref{thm: k-taxi} for the $k$-taxi problem is almost implied by the monotonicity of \cref{ex:taxiPotential} (the $k$-taxi potential), \cref{thm: reduction} and \cref{thm: general}. The only remaining issue is that our embedding algorithms require prior knowledge of $n$. We now show that we can handle the unknown-$n$ scenario using the ``guess-and-double'' technique. The proof for \cref{thm: k-server} ($k$-server) requires some technical desiderata and is deferred to \cref{appendix: kserver}.

\newcommand{\opt}{\mathrm{opt}}
\ktaxi*
\begin{proof}
We present the proof for general metric spaces, where our embedding provides distortion $O(\log^2 n)$. The case of bounded-dimensional normed spaces follows similarly, with all occurrences of $\log^2n_t$ replaced by $\log n_t$.

Let $\opt_t$ denote the optimal offline cost for serving the first $t$ requests.  Let $\seq r$ be the sequence revealed online. Recall $V(\seq r_{[0, t]})$ is the set of points among the initial locations and the first $t$ requests, and let $n_t = |V(\seq r_{[0, t]})|$. Thus, $n = |V(\seq r)|$. Our algorithm maintains two guesses throughout its execution: a guess $\zeta$ for the value of $\log^2n\cdot\opt$, and a guess $m$ for the total number of points. 

\paragraph{Phase Partitions.}
Initially (phase 1), we set $\zeta_1 = \min_{u \ne v} d(u, v)$ and $m_1 = 2$. Each phase $i$ ends at $t_i$, defined as the earliest time $t$ satisfying the condition 
\[
\zeta_i < \log^2 n_t\cdot \opt_t,
\]
after which phase $i+1$ immediately begins. When phase $i > 1$ starts, we update our guesses by doubling the previous value of $\zeta$ and an upper bound of $m$,
\[
\zeta_i = 2\zeta_{i-1}, \qquad
m_i = n_t^2.
\]
This ensures that at every time $t$ within phase $i$, we have $\sqrt{m_i} \le n_t \le m_i$.
The left inequality follows immediately from the definition of $m_i$. The right inequality holds because, if $n_t > m_i$, a phase transition would have already occurred: the condition $\log^2n_t\cdot\opt_t > 2\log^2(\sqrt{m_i})\cdot\opt_{t_{i-1}} = 2\zeta_{i-1} = \zeta_i$ would have triggered the start of a new phase.

\paragraph{The algorithm.} For each time $t$ in phase $i$, the algorithm maintains an embedding of $V(\seq r_{[0, t]})$ into an HST, assuming the total number of points is $m_i$, as guaranteed by \cref{thm: general}. Let $\Alg_i$ denote the algorithm obtained by simulating the $k$-taxi algorithm of~\cite{coester2019online} on the HST embedding maintained during phase $i$. Throughout phase $i$, the positions of taxis exactly follow those in $\Alg_i$. Note that the configuration of taxis may change significantly at phase transitions.

\paragraph{The Analysis.}
Suppose there are $p$ phases in total. When transitioning between phases $i$ and $i+1$ at time $t_i$, the algorithm incurs a cost of at most
\[
\mathrm{cost}_{\Alg_i}(\seq r_{[0,t_i]}) + \mathrm{cost}_{\Alg_{i+1}}(\seq r_{[0,t_i]}),
\]
as we could return to the initial configuration before switching algorithms. Within phase $i$, the cost incurred is identical to $\mathrm{cost}_{\Alg_i}$. Therefore, the total cost over all phases is at most $2\sum_{i=1}^{p}\mathrm{cost}_{\Alg_i}(\seq r_{[0,t_i]})$. By applying \cref{thm: reduction} and the guarantee in \cite{coester2019online}, we have for each phase $i$:
\[
\mathrm{cost}_{\Alg_i}(\seq r_{[0,t_i]}) \le O(2^k\log^2 m_i)\cdot\opt_{t_i} \le 2^k\zeta_i.
\]

Since the guesses satisfy $\zeta_i = 2\zeta_{i-1}$, summing over all phases yields the total cost is at most 
\[
2\sum_{i=1}^p \mathrm{cost}_{\Alg_i}(\seq r_{[0,t_i]}) \le 2^{k+1}
\sum_{i=1}^{p}\zeta_i \le 2^{k+2}\zeta_p \le O(2^k\log^2 m_p)\cdot\opt = O(2^k\log^2 n)\cdot\opt. \qedhere
\]
\end{proof}

\subsection{Application for Dynamic Embedding}
The application of dynamic embeddings can require problem-specific considerations, since the set of “relevant” points may vary for different problems. As an example, we revisit the $k$-taxi problem.

\begin{restatable}[]{theorem}{ktaxiEquivalent}
\label{thm: k-taxi equivalent}
The following two statements are equivalent:
\begin{enumerate}[itemsep=0em]
    \item There is an $f(k)$-competitive online algorithm for the $k$-taxi problem, for some function $f$.
    \item It is possible to maintain online an evolving set of at most $g(k)$ ``relevant'' points, such that the best \emph{offline} algorithm that must have taxis only at relevant points is an $h(k)$-approximation of the unrestricted optimal solution, for some functions $g$ and $h$.
\end{enumerate}
\end{restatable}
\begin{proof}
    We use $S_t$ to denote the set of relevant points at time $t$. The implication $1 \implies 2$ is trivial: simply set $S_t$ to be the locations of the online taxis. This yields $g(k) = k$ and $h(k) = f(k)$.

    For the implication $2 \implies 1$, we employ our monotone embedding in \cref{thm: general dynamic}. Suppose the configuration of the online algorithm at time $t$ is $C_t$. When a new request $(x_{t+1}, y_{t+1})$ arrives, we embed the points $C_t \cup S_t \cup S_{t+1} \cup \{x_{t+1}, y_{t+1}\}$ 
    into an HST, and simulate the online algorithm from~\cite{coester2019online} on this HST embedding. After serving the request, we remove all points except those in $C_{t+1} \cup S_{t+1}$. We claim that this algorithm achieves a competitive ratio of
    \[
        f(k) = O(2^k h(k) g(k) \log (g(k))).
    \]

    Consider the optimal offline solution $\seq b$ that restricts its taxis to locations in $S_t$ at every time $t$. Recall from~\cite{coester2019online} that the potential function is defined as $(2^{k}-1)$ times the minimum-weight matching between the offline and online taxi locations. Therefore, monotonicity (among the alive points) ensures that updating the embedding never increases this potential. Moreover, since $|C_t \cup S_t \cup S_{t+1} \cup \{x_{t+1}, y_{t+1}\}| \le 2g(k) + k + 2$, the embedding distortion at each step is at most $O(g(k)\log(g(k)))$. A similar analysis to that in \cref{thm: reduction} shows that the algorithm is 
    $O(2^k g(k)\log(g(k)))$-competitive against the restricted offline solution $\seq b$, which implies the overall competitive ratio of $O(2^k h(k) g(k)\log(g(k)))$.
\end{proof}

\cref{thm: k-taxi equivalent} is a speculative theorem that relates the competitive ratio of the $k$-taxi problem to the maintenance of a relevant set of points. We further provide a partial characterization based on the size of the support of the work function. 

\begin{definition}[Work Function Width of $k$-taxi]
Let $w_t(C)$ be the minimum cost of serving the first $t$ requests in a $k$-taxi instance and ending in configuration $C$. A configuration $C$ is in the \emph{support} at time $t$ if there is no distinct configuration $C'$ with $w_t(C) = w_t(C') + d(C',C)$. The \emph{work function width} of an instance is the maximum number of points that appear in these support sets over every time $t$.
\end{definition}

\begin{corollary}
There is an $O(2^kl\log(l))$-competitive algorithm for any $k$-taxi instance whose work function width is $l$.
\end{corollary}

\begin{proof}
An optimal offline solution can remain in a support configuration at each step, because if it were in some $C$ outside the support, it would be just as good to be at $C' \neq C$ with $w_t(C) = w_t(C') + d(C', C)$. Also, the online algorithm can maintain the support. The theorem follows from \cref{thm: k-taxi equivalent}.
\end{proof}

\section{Conclusion and Discussion}
\label{sec: conclusion}
We present a new framework, online monotone metric embeddings, which allows the embedding to evolve over time, provided that distances between points do not increase. For embeddings into HSTs, we establish $O(\log^2 n)$ distortion from general metrics and tight $O(\log n)$ distortion from $O(1)$-dimensional normed spaces. We also discuss dynamic monotone embeddings and present an $O(l \log l)$ upper bound and $\Omega(l)$ lower bound for the distortion. Finally, we discuss the conditions for an algorithm to be combined with our embedding, and illustrate some applications.

While we use monotone potential functions in~\cref{thm: reduction} as a systematic way to certify compatibility with particular online algorithms, our embedding framework may be applicable more broadly. It suffices to inspect an algorithm’s proof and verify that the argument remains valid under monotone distance updates.
In this way, future work may be able to directly plug our embeddings into their analyses.
We highlight some additional future directions below.

\begin{enumerate}[itemsep=0em]
    \item What is the tight bound for the distortion of online monotone embeddings into HSTs? Deterministically, our algorithm is optimal. For probabilistic embeddings, the only known lower bound is the $\Omega(\log n)$ offline lower bound. Our algorithms match this on constant-dimensional normed spaces, but leave a quadratic gap in the general case.
    \item Are there interesting results for online monotone embedding into metrics other than trees? We mainly focus on HSTs since many online problems have competitive algorithms on HSTs. However, embedding into other metrics is also of interest.
    \item Dynamic embedding with limitations on the type of recourse is a field with major potential for many online problems, of which we have only scratched the surface. Dynamic embeddings have been a main candidate for getting a $\polylog(k)$-competitive algorithm for the randomized $k$-server problem. Our $\Omega(l)$ lower bound imposes limitations on certain approaches. As another example, our \cref{thm: k-taxi equivalent} directly suggests pathways for the $k$-taxi problem.
    \item Are there similar characterizations for other online problems with monotone recourse? For example, can improved results be obtained for online matching with monotone recourse, where agents are only willing to change their partner if they prefer the new partner?
\end{enumerate}
\paragraph{Acknowledgements.} We thank Michael Mitzenmacher for his insightful comments on earlier versions of this paper, and the anonymous reviewers for their constructive feedback. YH is supported by NSF grant CNS-2107078. Part of the research was performed while YH was affiliated with the University of Oxford. CC is funded by the European Union (ERC, CCOO, 101165139). Views and opinions expressed are however those of the author(s) only and do not necessarily reflect those of the European Union or the European Research Council. Neither the European Union nor the granting authority can be held responsible for them.

\printbibliography

\newpage
\appendix
\renewcommand{\thetable}{\Alph{table}}
\section{Lower Bounds for Strict Embeddings}
\label{sec: lower bound without recourse}
We show lower bounds for online strict embeddings, even for $l=3$. 

In the following parts, we sometimes refer to the notion of ultrametric. An ultrametric is a metric space such that for all points $u, v, w$, \[d(u, v) \leq \max (d(u, w), d(v, w)).\] 
Note that all HSTs are ultrametrics and an ultrametric is a $1$-HST.

\begin{theorem}
\label{thm: online_lower}
    There exists an adversary with width $l=3$ against which any deterministic online embedding into
    HSTs has distortion $2^{n-2}$.
\end{theorem}
\begin{proof}
    An adversary reveals points on a line of length 1 as follows: initially, $v_1, v_2$ at positions 0 and 1, respectively. After placing $i$ points, the adversary maintains a pair $u, v$ where $d(u, v) \leq 2^{2-i}$ but $d_T(u, v) \geq d_T(v_1, v_2)$ for all $i \geq 2$. Initially, $u = v_1$ and $v = v_2$. The adversary selects $x$ at the median of $u, v$. Given that HST is an ultrametric, $\max(d_T(x, v), d_T(x, u)) \geq d_T(u, v)$. Assuming without loss of generality that $d_T(x, u) \geq d_T(x, v)$, we update $u, v$ to $x, u$ and continue. This results in a distortion of at least:
    \[\frac{d_T(v_1, v_2)} {2^{2-n}} \frac{1}{d_T(v_1, v_2)} \geq 2^{n-2}.\]

    Note that the adversary only needs to keep $u, v, x$ at any time, with a width of $3$.
\end{proof}

The following lower bound construction is from \cite{Indyk10online}, and we show that the lower bound preserves even for $l=3$.
\begin{theorem}
\label{thm: online_lower_random}
    There exists a non-adaptive adversary with width $l=3$ against which any probabilistic online embedding into HSTs has distortion $\Omega(n)$.
\end{theorem}

\begin{proof}
    We use Yao's min-max principle and specify a distribution over update sequences and a pair of points on which the expansion is measured, and show that no deterministic algorithm achieves expected distortion better than $\Omega(n)$. Our sequence will contain $n$ points, and since there are departures, the length will be longer than $n$.
    
    The adversary reveals points on a line as follows. Initially $v_1=0, v_2=1, v_3=1/2$. After that $v_{i+1} = v_{i} + b_i 2^{-i+2}$ where $b_i \in \{-1, 1\}$ is chosen uniformly at random. 
    
    It is easy to see that for each $i \geq 3$, there exist points $l_i, r_i \in X$ such that $l_i = v_i - 2^{-i+2}$, $r_i = v_i + 2^{-i+2}$, and $\{v_1, \dots, v_i\} \cap [l_i, r_i] = \{l_i, v_i, r_i\}$. Moreover, for each $i \in \{3, \dots, n-1\}$, there uniquely exists $y_i \in \{l_i, r_i\}$ such that $\{v_{i+1}, \dots, v_n\} \subset [\min\{v_i, y_i\}, \max\{v_i, y_i\}]$. After placing $v_i$, remove anything other than $v_i$ and $y_i$. Thus, the width is $3$. By our construction, $l_i, r_i \in L_i$.

    We use $t_i$ to denote the time the adversary places $v_i$. Since HSTs are ultrametrics, there exists $z_i \in \{l_i, r_i\}$ such that $d_{t_i}(v_i, z_i) \geq d_{t_{i}}(l_i, r_i)$. We denote by $D_i$ the event that $d_{t_i}(v_i, y_i) \ge d_{t_{i-1}}(v_{i-1}, y_{i-1})$. We thus have 
    \[\P\left(d_{t_i}(v_i, y_i) \ge d_{t_{i-1}}(v_{i-1}, y_{i-1}) \mid \bigwedge_{j < i} d_{t_j}(v_j, y_j) \ge d_{t_{j-1}}(v_{j-1}, y_{j-1})\right) \ge 1/2\] as we pick $z_i$ with probability $1/2$. Applying induction gives $\P(d_{t_i}(v_i, y_i) \ge d_{t_{i-j}}(v_{i-j}, y_{i-j})) \ge 2^{-j}$

    Thus, suppose the contraction is $\lambda_c$,
    \begin{align*}
    \E[d_{t_{n-1}}(v_{n-1}, y_{n-1})] &\ge \sum_{i=3}^{n-1} d_{t_{n-1}}(v_{n-1}, y_{n-1}) 2^{i-n+1} \\&\ge \sum_{i=3}^{n-1} \frac{2^{-i+2}}{\lambda_c} 2^{i-n+1} = \frac {\Omega(n2^{-n})} {\lambda_c} = \frac {\Omega(n) d(v_{n-1}, y_{n-1})} {\lambda_c}
    \end{align*}
    Therefore, the distortion is $\Omega(n)$, even when $l=3$.
\end{proof}

\section{From Partitions to HST Embeddings}
\label{appendix: tree building}
In this section, we discuss how to convert partitions into an HST embedding and provide the previously omitted proofs.

\begin{definition}[Induced HST Embedding]
    The \emph{refinement} of two partitions $C_1, C_2$ is defined as the partition 
    \[
        \br{S_1 \cap S_2 \mid S_1 \in C_1, S_2 \in C_2} \setminus \br{\varnothing}.
    \] 
    
    For a set of scales $\br{s_j}$ and partitions $\br{C^j}$ for every $j \in \mathbb{Z}$, we define the \emph{induced HST embedding} as follows: First, we refine each partition at level $j$ with all partitions at levels above, turning the partitions into a nested hierarchy. 
    
    We restrict our attention to levels at or above $Y$, the largest level for which $C^Y$ consists entirely of singleton sets. If no such $Y$ exists, the induced embedding is undefined. For each level $j$ and every cluster $S \in C^j$, we create a vertex $v_S$ with weight $\varphi(v_S) = s_j$. The vertices at level $Y$ form the leaves and thus have $\varphi(v_S) = 0$. Each vertex $v_S$ at level $j$ is made a child of the vertex corresponding to the cluster at level $j+1$ that encompasses $S$. Finally, each point $u$ is mapped to the vertex corresponding to the singleton cluster $\br{u}$ at level $Y$.
\end{definition}

We are now ready to prove~\cref{lemma: hierarchy}.

\Construction*
\begin{proof}
    Consider the set of scales given by $s_j = 2^j$. For every time $t$, we construct the induced HST embedding based on scales $\br{s_j}$ and partitions $\{C^j_t\}$. Since the partitions at each level have monotonicity, the resulting refined partitions are also monotone. Recall that points in the same cluster belong to the same subtree in the HST, implying distances never increase, as the least common ancestor of two points never moves up the hierarchy.

    We now bound the distortion. Fix two points $u, w$. Let $A$ be the largest integer $j$ such that $2^j < d(u,w)$ or $\delta_t^{(j)} \neq 0$, and define $B = \lceil \log_2 d(u,w) \rceil$. If no such $A$ exists, then $\sum_j \delta_t^{(j)}$ is infinite, making the claim trivial. Let $R^j(u, w)$ denote the event that $C^j(u) = C^j(w)$, $D^j(u, w) = \P(\neg R^j(u, w) \mid R^{j+1}(u,w))$ and $E^j(u, w) = \E[d_T(u, w) \mid R^{j+1}(u, w)]$. If $R^{j+1}(u,w)$ is impossible, set $D^j(u,w)=E^j(u,w)=0$.

    Note $R^{A+1}(u,w)$ always holds, so $\E[d_T(u,w)] = E^A(u,w)$. Moreover, $\neg R^{B-1}(u,w)$ always holds, so $d_T(u,w) \ge d(u,w)$, ensuring the embedding is non-contractive.

    Let $l = \lfloor \log_2(d(u,w)/\varepsilon) \rfloor$ be the highest level for which $d(u,w) \ge \varepsilon 2^l$. Note $l - B \le \lceil \log(1/\varepsilon)\rceil$. Observe that $D^j(u,w)$ is the probability that $C^j$ splits $u$ and $w$, given they remain together at all levels above $j$. Thus, despite the refinement, we have $D^j(u,w)\le \delta_t^{(j)}$ for all $j \ge B$. We have
    \begin{align*}
        E^j(u,w) 
        &= D^j(u,w)\cdot 2^j + (1 - D^j(u,w))E^{j-1}(u,w) \\
        &\le \frac{\delta_t^{(j)} \cdot d(u,w)}{2^j}\cdot 2^j + E^{j-1}(u,w) \\
        &= \delta_t^{(j)} \cdot d(u,w) + E^{j-1}(u,w).
    \end{align*}

    For levels $j > l$, we have the stronger bound $D^j(u,w) \le \gamma \delta_t^{(j)}$, so similarly:
    \[
        E^j(u,w) \le \gamma \delta_t^{(j)} \cdot d(u,w) + E^{j-1}(u,w).
    \]

    Combining these results, we have:
    \begin{align*}
        \E[d_T(u,w)] 
        &= E^A(u,w) 
        \le \sum_{j=l+1}^{A}\gamma\delta_t^{(j)} \cdot d(u,w) + E^l(u,w) \\
        &\le \left(\gamma \sum_{j\in\mathbb{Z}}\delta_t^{(j)} + \sum_{j=B}^{l}\delta_t^{(j)}\right) d(u,w) + E^{B-1}(u,w) \\
        &\le \left(\gamma\sum_{j\in\mathbb{Z}}\delta_t^{(j)} + \log\frac{1}{\varepsilon}\max_{j\in\mathbb{Z}}\delta_t^{(j)} + 1\right)d(u,w).
    \end{align*}

    Thus, the claimed distortion bound follows.
\end{proof}

Finally, we provide the proof of \cref{lemma: relevant scales}, which follows \cite{relevantscales}. We use the following lemma, which says that for a set with size $n$, the number of relevant scales from the sum of subsets with some offset is bounded by $2n$. We write $[n] = \{0, \ldots, n-1\}$.
\begin{lemma}
\label{lemma: count}
    Let $\seq{a} = a_0, \ldots, a_{n-1}$ be $n$ real numbers, then $|S(\seq a)| \le 2n$, where \[S(\seq{a})=\left\{k \in \Z \mid \exists \varnothing\subsetneq I\subseteq [n] , 0 \le x < a_0: x+\sum_{i\in I} a_i \in [2^{k-1}, 2^k)\right\}.\]
\end{lemma}
\begin{proof}
    We assume without loss of generality that $a_i \le a_{i+1}$ for all $i$. We induct on the number of elements $n$. For $n = 1$ the number of relevant scales is $2$. We now show inductive steps. If for all $j$, $a_j \le a_0 + \sum_{i < j} a_i$, then $a_j \le 2^j a_0$ so $\sum_{i=0}^{n-1} a_i \le (2^n-1)a_0$. Thus, for $I\subseteq [n] \neq \varnothing, 0 \le x < a_0$, we have $x+\sum_{i\in I} a_i \in [d_0, 2^n d_0)$ and $|S(\seq a)|\le n+1$. 
    
    Otherwise, suppose $j$ is the smallest index such that $a_j > a_0 +\sum_{i < j} a_i$. 
    Let
    \begin{align*}
        S_1&=\left\{k \in \Z \mid \exists \varnothing\subsetneq I\subseteq [j] , 0 \le x < a_0: x+\sum_{i\in I} a_i \in [2^{k-1}, 2^k)\right\}\\
        S_2&=\left\{k \in \Z \mid \exists I\subseteq [n], I \setminus [j] \neq \varnothing, 0 \le x < a_0: x+\sum_{i\in I} a_i \in [2^{k-1}, 2^k)\right\}
    \end{align*}
    be the relevant scales depending on whether an element larger than $a_j$ is used.
    Clearly $S(\seq a) = S_1 \cup S_2$. Let $\seq b = a_0, \ldots, a_{j-1}$ and $\seq c = a_j, \ldots, a_{n-1}$. Note that $S_1 \subseteq S(\seq b)$ and $S_2 \subseteq S(\seq c)$, where the first is by definition and the second is by $x+\sum_{i\in I} a_i = x + \sum_{i \in I\cap[j]} a_i + \sum_{i \in I \setminus [j]}a_i$ and take $x' = x + \sum_{i \in I\cap[j]} a_i < a_j$. By inductive hypothesis, $|S_1| \le |S(\seq b)| \le 2j$ and $|S_2| \leq |S(\seq c)|\le 2(n-j)$, so $|S(\seq a)| \leq 2n$.
\end{proof}

\begin{proof}[Proof of \cref{lemma: relevant scales}]
    We first prove the case without $\varepsilon$, i.e.,
    \[
    |\{j \in \mathbb{Z} \mid \exists u,v \in V: d(u,v)\in[2^{j-1},2^j)\}| = O(n).
    \]
    
    Treat the metric as a complete graph with edge weights being the distance between vertices, and build a minimum spanning tree on this graph. Let $d_1\le \ldots\le d_{n-1}$ denote the weight of the edges in the spanning tree. 
    
    For all points $u,v$, let $e_0\le \ldots\le e_{l-1}$ be the weight of edges on the path in the tree from $u$ to $v$, then for all $k$, $d(u, v) \in [2^{k-1}, 2^k) \implies \exists I \subseteq [l]: \sum_{j \in I} e_j \in [2^{k-1}, 2^k)$. To see this, note $e_{l-1} \le d(u, v) \le \sum_{j=0}^{l-1} e_j$ where the first inequality is by the property of a minimum spanning tree and the second is by triangle inequality. On the other hand, the set $\{s_i\}$ where $s_i = \sum_{j=l-i}^{l-1} e_j$ the sum of largest $i$ elements covers all interval of form $[2^{k-1}, 2^k)$ in $[e_{l-1}, \sum_{j=0}^{l-1} e_j]$ since $e_j \le e_{j+1} \le s_{j+1}$.

    Thus, $\exists u, v: d(u, v) \in [2^{k-1}, 2^k) \implies \exists I \subseteq [n-1]: \sum_{i \in I} d_i \in [2^{k-1}, 2^k)$. Applying \cref{lemma: count} gives the desired claim.

    When $\varepsilon < 1$, note that every relevant scale without $\varepsilon$ implies $O(\log (1/\varepsilon))$ scales to be relevant when $\varepsilon < 1$, which completes the proof.
\end{proof}

\section{Bounded Smooth Partitions}
\label{section: appendix general dynamic}
Here we include the omitted proof for \cref{clm: bartal padding}.
\begin{proof}[Proof of \cref{clm: bartal padding}]
    Consider a fixed time, and fixed points $u$ and $w$ in the current alive set. Let $B_1,B_2,\dots$ be the components of the current partition and denote by $(c_j, r_j, t_j)$ the triple associated with $B_j$, sorted in increasing order of creation time $t_j$. Let us define the following events.
    
    \begin{enumerate}[itemsep=0em]
        \item $X_j$ is the event that \emph{exactly one} of $u, w$ is within distance $r_j$ of $c_j$.
        \item $N_j$ is the event that \emph{none} of $u, w$ is within distance $r_j$ of $c_j$, and $A_j = \bigwedge_{l < j} N_l$ the event that $u, w$ are still \emph{available} before $j$th component.
        \item $F_j=X_j\land A_j$ is the event that $B_j$ is the \emph{first} cluster that exactly one of $u, w$ belongs to.
    \end{enumerate}
    
    Let $k' \le l$ be the largest number such that having $k'$ alive components is not an impossible event. We want to bound $\P(C_t(u) \neq C_t(w)) = \sum_{j=1}^{k'} \P(F_j)$. For each $j \le k'$, let random variable $\seq {S_j} \in (V_t)^j$ denote the sequence of centers of alive components $B_1, \ldots, B_j$. Note here we used $V_t$, the set of points seen so far, instead of $L_t$ the alive points, since the center of a component might not be alive. We perform a reverse induction to show that for all $j$ and all $\seq {s_{j-1}}$ such that $\P(A_j \wedge \seq{S_{j-1}} = \seq {s_{j-1}}) > 0$, 
    \begin{align*}
        \sum_{i \ge j}\P\l F_i \mid A_j \wedge \seq{S_{j-1}} = \seq {s_{j-1}}\r \leq \l\sum_{i \ge j}\frac {1} {\chi_i^2-1} + 1\r\frac {32d(u, w) \log \chi_l} s.
    \end{align*} 

    We write $E_j = \l\sum_{i \ge j}\frac {1} {\chi_i^2-1} + 1\r\frac {32d(u, w) \log \chi_l} s$ for brevity.
    Note that for $j=k'+1$, the above holds trivially. Suppose it holds for $j+1$. Fix an $\seq {s_{j-1}}$. We call a point $v \in V_t$ valid if $\P(\seq {S_j} = \seq {s_{j-1}}\cdot v \mid H_j \wedge \seq {S_{j-1}} = \seq {s_{j-1}}) > 0$, where $\seq{s_{j-1}} \cdot v$ is the sequence concatenating $v$ to $\seq {s_{j-1}}$.

    Fix an arbitrary valid $v$. Without loss of generality assume $d(v, u) \le d(v, w)$. Let $\seq {s_j} = \seq {s_{j-1}}v$ and
    $\tilde d(u, w) = \min (\max(d(u, w), s/16), s/8)$ truncating the distance
    to the domain of $p$. It is worth noting that $\tilde d(v, w) -
    \tilde d(v, u) \leq d(u, w)$.

    We have
\begin{align}
  \P\l X_j \mid A_j \wedge \seq{S_{j}} = \seq{s_{j}}\r \nonumber
    &= \int_{\tilde d(v, u)}^{\tilde d(v, w)} p(z) dz \\
    &\leq \frac{\chi_j^{2}}{1-\chi_j^{-2}}
          \cdot \frac{32(\tilde d(v, w) - \tilde d(v, u))\log\chi_j}{s}
          e^{-\frac{32\tilde d(v, u)\log\chi_j}{s}} \nonumber\\
    &\leq \frac{\chi_j^{2}}{1-\chi_j^{-2}}
          \cdot \frac{32d(u, w)\log\chi_j}{s}
          e^{-\frac{32\tilde d(v, u)\log\chi_j}{s}}, \label{eq: single split}\\
    \P\l N_j \mid A_j \wedge \seq{S_{j}} = \seq{s_{j}}\r&=\int_{s/16}^{\tilde d(v, u)} p(z) dz
    = \frac{\chi_j^{2}}{1-\chi_j^{-2}}
       \left(
         \chi_j^{-2} - e^{-\frac{32\tilde d(v, u)\log\chi_j}{s}}
       \right).\nonumber
\end{align}

We now have
\begin{align*}
&\quad\;\sum_{i \ge j}\P\l F_i \mid  A_j \wedge \seq{S_{j}} = \seq {s_{j}}\r \\
&= \P\l X_j\mid A_j \wedge \seq{S_{j}} = \seq{s_{j}}\r  + \P\l N_j\mid A_j \wedge \seq{S_{j}} = \seq{s_{j}}\r
   \sum_{i \ge j+1}\P\l F_i\mid A_{j+1} \wedge \seq{S_{j}} = \seq{s_{j}}\r \\
&\leq 
  \frac{\chi_j^{2}}{1-\chi_j^{-2}}
     \cdot \frac{32d(u, w)\log\chi_j}{s}
     e^{-\frac{32\tilde d(v, u)\log\chi_j}{s}} + \frac{\chi_j^{2}}{1-\chi_j^{-2}}
     \left(
       \chi_j^{-2} - e^{-\frac{32\tilde d(v, u)\log\chi_j}{s}}
     \right) E_{j+1} \\
&= \frac{\chi_j^{2}}{1-\chi_j^{-2}}
    \left[
      \frac{32d(u, w)\log\chi_j}{s}
      e^{-\frac{32\tilde d(v, u)\log\chi_j}{s}}
      - e^{-\frac{32\tilde d(v, u)\log\chi_j}{s}} E_{j+1}
      + \chi_j^{-2} E_{j+1}
    \right] \\
&= \frac{\chi_j^{2}}{1-\chi_j^{-2}}
     \left(
       e^{-\frac{32\tilde d(v, u)\log\chi_j}{s}}
       \left[
         \frac{32d(u, w)\log\chi_j}{s}
         - E_{j+1}
       \right]
       + \chi_j^{-2} E_{j+1}
     \right)
\end{align*}
Since $e^{-\frac{32\tilde d(v, u) \log \chi_j}{s}} \geq \chi_j^{-2}$ and $E_{j+1} \geq \frac{32d(u, w)\log\chi_j}{s}$,
\begin{align*}
  \sum_{i \ge j}\P\l F_i \mid  A_j \wedge \seq{S_{j}} = \seq {s_{j}}\r
  &\leq \frac{\chi_j^{2}}{1-\chi_j^{-2}}
       \left(
         \chi_j^{-2}
         \left[
           \frac{32d(u, w)\log\chi_j}{s} - E_{j+1}
         \right]
         + E_{j+1}
       \right) \\
  &= E_{j+1} + \frac{1}{\chi_j^{2}-1}
      \frac{32d(u, w)\log\chi_j}{s} \le E_j
\end{align*}

This holds for all valid $v$, and in the case that it is impossible to have the $j$th component, the probability is $0$. Therefore, we have
\begin{align*}
    \sum_{i \ge j}\P( F_i \mid  A_j \wedge \seq{S_{j-1}} = \seq {s_{j-1}}) \leq E_j, &&
\sum_{j\ge1} F_j \le E_1 = \l\sum_{i \ge 1}\frac {1} {\chi_i^2-1} + 1\r\frac {32d(u, w) \log \chi_l} s.
\end{align*}
    Recall $\chi_j = 2j$, so $\sum_i 1/(\chi_i^2 - 1) \le 2\sum \chi_i^{-2} = O(1)$ and $\log \chi_l = O(\log l)$. We conclude $\P(C_t(u) \ne C_t(w)) \le O(\log l) \cdot \frac{d(u, w)}{s}$.
\end{proof}

\section{Normed Spaces}
\label{appendix: normed}
Here we provide embedding results for normed spaces. At a high level, the argument closely follows the approach for general metric spaces, with the primary difference being the construction of the base partition $C$. We first construct an $O(1)$-smooth probabilistic partition for the real line. We then extend it to $\ell_\infty^D$ by independently constructing partitions in each dimension and taking their refinement, yielding an $O(D)$-smooth partition. Finally, we generalize the results to arbitrary normed spaces in $\mathbb{R}^D$, utilizing the equivalence of norms up to a factor of $D$.

\paragraph{Smooth Probabilistic Partitions on the Line.}
Without loss of generality, assume each point $v_t$ has a coordinate $x_{v_t}$ on the real line. Each component corresponds to an interval. We fix a global set of intervals at time $0$, denoted as $C$, and define $C_t$ as the subset of intervals containing points in $V_t$ at time $t$.

Let $r = s/3$. For each integer $k$, define intervals $I_k = [k r, (k+1) r)$ and select random points (cutting points) $z_k$ uniformly in the interval $[k r, (k+1/2) r]$. The intervals between consecutive cutting points form our components, namely, each component is $[z_k, z_{k+1})$. A point $u$ belongs to the component containing its coordinate $x_u$.

\begin{claim}
\label{clm: padding}
If $d(u,w) \le s$, then for all times $t$,
\[
\P(C_t(u) \neq C_t(w)) \leq \frac{6\, d(u, w)}{s}.
\]
\end{claim}

\begin{proof}
If $d(u, w) \ge s/6$, the claim trivially holds as the probability is at most $1$. Assume otherwise and without loss of generality that $x_u \le x_w$. 

If $u,w$ belong to the same interval $I_k$, then by construction,
\[
\P(C_t(u) \neq C_t(w)) 
= \frac{|(x_u, x_w] \cap [k r, (k+1/2) r]|}{r/2}
\leq \frac{6\, d(u,w)}{s}.
\]

If instead $u \in I_{k-1}$ and $w \in I_k$, we similarly have:
\[
\P(C_t(u) \neq C_t(w)) 
= \frac{|[k r, x_w) \cap [k r, (k+1/2) r]|}{r/2}
\leq \frac{6\, d(u,w)}{s}. \qedhere
\]
\end{proof}

\begin{theorem}
\label{thm: line partition}
For every $n \in \mathbb{N}$, $s > 0$, and $\varepsilon \le n^{-1}$, given knowledge of $n$, there is a probabilistic online $s$-bounded monotone partition of up to $n$ points on the line that is $(O(1), \varepsilon, O(n\varepsilon))$-smooth at every time $t$. Moreover, if no pair of points in $V_t$ has distance within $[\varepsilon s, s]$, the partition is $0$-smooth at all times.
\end{theorem}

\begin{proof}
We follow the general outline from \cref{sec: overview}, using the partition defined above for the base partition $C$. In contrast to the general metric case, allowing just one merge would be insufficient here, because our construction does not ensure $0$-smoothness; random cutting points may fall between two close points regardless of point placement. 

Therefore, we allow each component to merge at most once with an adjacent component. Initially, each component forms a singleton cluster in $G$. Upon arrival of a new point $v_t$, if there exists $u$ with $d(u,v_t) \le \varepsilon s$ and $G_t(C_t(u)) \neq G_t(C_t(v_t))$, we attempt to merge these two clusters, provided neither cluster has merged previously. Because each interval's length is at most $1.5r$, the diameter after merging adjacent intervals remains at most $3r = s$, thus preserving $s$-boundedness.

\begin{claim}
For any time $t$ and points $u,w \in V_t$ with $d(u,w)\le \varepsilon s$,
\[
\P(G_t(C_t(u)) \ne G_t(C_t(w)) \mid C_t(u)\ne C_t(w)) \le O(n\varepsilon).
\]
\end{claim}

\begin{proof}
Conditioned on $C_t(u) \neq C_t(w)$, there exists a cutting point $z$ between $u$ and $w$. A merge attempt fails only if one of the adjacent cutting points splits another close pair, thus already using up the allowed merge. By independence of cutting points and applying Claim~\ref{clm: padding}, this conditional probability is at most
\[
\sum_{\substack{v, v' \text{ adjacent}\\d(v, v') \le \varepsilon s}} \P(C_t(v)\neq C_t(v')) \le O(n \varepsilon).\qedhere
\]
\end{proof}

\begin{claim}
If no pair of points in $V_t$ has distance within $[\varepsilon s,s]$, then for all $u,w$ with $d(u,w)\le s$, we have $G_t(C_t(u)) = G_t(C_t(w))$.
\end{claim}

\begin{proof}
Points naturally form groups with diameter at most $\varepsilon s$, and these clusters are separated by more than $s$. Hence, merge attempts within each group succeed unobstructed, guaranteeing the claim.
\end{proof}
This completes the proof of \cref{thm: line partition}.
\end{proof}

Taking $\varepsilon = n^{-3}$ and applying \cref{lemma: hierarchy} gives the embedding for the line.

\begin{restatable}[Line Incremental]{theorem}{Line}
\label{thm: line embedding}
For every $n\in \mathbb N$, there exists a probabilistic online monotone embedding of up to $n$ points from the line metric into HSTs with distortion $O(\log n)$.
\end{restatable}

Constructing partitions independently per dimension and taking their refinement directly implies the result for $\ell_\infty$ metrics.

\begin{theorem}[$\ell_\infty^D$ Incremental]
\label{thm: l_inf}
For every $n \in \mathbb N$, there is a probabilistic online monotone embedding of up to $n$ points from $\ell_\infty^D$ into HSTs with distortion $O(D\log(nD))$.
\end{theorem}
\begin{proof}
    For each dimension $i = 1, 2, \dots, D$, let $d^i(u,w)$ denote the distance between $u$ and $w$ at the $i$th dimension and using $d^i$ as the metric, construct an online probabilistic monotone partition $\{P^i_t\}$ that is $(\delta_t, \varepsilon, \gamma)$-smooth at time $t$.\footnote{Previously, we used the superscript for the level; here it is used for the dimension since we are now focusing on a partition on a single level.} At each time, the partition $P_t$ is defined such that $P_t(u) = P_t(w)$ if $P^i_t(u) = P^i_t(w)$ holds for all $1 \le i \le D$. Since each partition $P^i_t$ is $s$-bounded and satisfies monotonicity, their refinement has these properties as well.

    Now, if $d(u,w) \le s$ in $\ell_\infty^D$, we have $d^i(u,w) \le s$ for every dimension $i$. Therefore
    \[
    \P(P_t(u)\neq P_t(w)) \leq \sum_{i=1}^D\P(P_t^i(u) \neq P_t^i(w)) \leq \sum_{i=1}^D \frac {d^i(u,w)} s \cdot \delta \leq D\delta\frac {d(u,w)} s.
    \]

    Similarly, if $d(u, w) \leq \varepsilon s$, then $\P(P_t(u)\neq P_t(w)) \leq D\delta\gamma\frac {d(u,w)} s$.

    Applying \cref{thm: line partition} to each dimension and setting $\varepsilon = n^{-3}D^{-1}$ yields a partition for points in $\ell_\infty^D$ that is $(O(D), n^{-3} D^{-2}, O(n^{-2} D^{-2}))$-smooth. Observe that a partition is $0$-smooth if the partition in each dimension is $0$-smooth. Hence, the number of scales that are not $0$-smooth becomes $O(Dn \log (Dn))$. \cref{lemma: hierarchy} then gives a distortion of $O(D \log (nD))$.
\end{proof}

Similarly, we can modify the partition construction $C_t$ in the fully dynamic setting.

\begin{theorem}[$\ell_\infty^D$ Dynamic]
There exists a probabilistic online monotone embedding from $\ell_\infty^D$ into HSTs with distortion $O(Dl)$, where $l$ is the width of the sequence.
\end{theorem}

\begin{proof}
We construct $C_t$ by independently forming the $O(1)$-smooth partitions described earlier for each dimension, but this time with each component's diameter capped at $s/4$ (achieved by setting $r = s/6$ instead of $s/3$). We then define $C_t$ as the refinement of these $D$ partitions, resulting in an $O(D)$-smooth partition overall.

The merging strategy is the same as in \cref{sec: general dynamic}: whenever two points $u,w \in L_t$ satisfy $d(u,w)\le s/6$ and $C_t(u)\ne C_t(w)$, we attempt to merge their respective components. The same argument as in \cref{lemma: zero smooth} shows that if no pair of alive points has distance in $[s/6,s]$, the resulting partition is $0$-smooth at that scale.

By applying \cref{lemma: hierarchy}, the total distortion equals $\sum_j \delta^{(j)} = O(Dl)$, as required.
\end{proof}

To generalize further, we use the equivalence of norms in $\R^D$:

\begin{restatable}[Equivalent Norms]{lemma}{norms}
\label{lem: equivalent norms}
For any norm $\|\cdot\|$ in $\R^D$, there is a linear map $T$ such that for all $x \in \mathbb{R}^D$,
\[
\|x\| \le \|T x\|_\infty \le D\|x\|.
\]
\end{restatable}

Thus, embeddings from $\ell_\infty^D$ imply embeddings for any normed space in $\mathbb{R}^D$, completing the proof for the parts regarding normed spaces in \cref{thm: general} and \cref{thm: general dynamic}.

\section{Embeddings with Contractions}\label{app:contractions}
In our main results, we have focused on non-contractive embeddings. However, to handle embeddings without prior knowledge of $n$ (\cref{appendix: unknown n}) and discuss deterministic algorithms (\cref{appendix: deterministic parent}), we need embeddings that potentially involve contractions.

\begin{definition}[Online Monotone Embedding, Contraction Allowed]
\label{def: embedding2}
Let $(X, d_X)$ be a metric space and let $\mathcal{M}$ be a family of metric spaces. A \emph{(deterministic) online monotone embedding} from $(X, d_X)$ into $\mathcal{M}$ takes inputs from an update sequence $\sigma$ of length $n$ one by one, and upon receiving $\sigma_t$, outputs a metric $d_t$ satisfying the following conditions:
\begin{enumerate}[itemsep=0em]
    \item $M_t = (L_t, d_t) \in \mathcal{M}$;
    \item $d_{t}$ is dominated by $d_{t-1}$ on $L_{t-1} \cap L_t$: for all $u,v \in L_{t-1} \cap L_t$, 
    $d_{t}(u,v) \le d_{t-1}(u,v)$.
\end{enumerate}
Such an embedding has \emph{contraction} $\lambda_c$ if for every update sequence $\sigma$, time $t$, and points $u, v \in L_t$,
    \[
    \lambda_c \cdot d_t(u, v) \geq d_X(u, v),
    \]
    and it has \emph{expansion} $\lambda_e$ if for every update sequence $\sigma$, time $t$, and points $u, v \in L_t$, 
    \[
    d_t(u, v) \leq \lambda_e \cdot d_X(u, v).
    \]
    A probabilistic monotone embedding is a distribution over deterministic ones. It has \emph{contraction} $\lambda_c$ if every deterministic embedding in its support has contraction $\lambda_c$, and it has \emph{expansion} $\lambda_e$ if for every update sequence $\sigma$, time $t$, and points $u, v \in L_t$,
    \[
    \mathbb{E}[d_t(u, v)] \leq \lambda_e \cdot d_X(u, v).
    \]

    An embedding is called \emph{non-expansive} (resp., \emph{non-contractive}) if its expansion (resp., contraction) is at most $1$ in every realization. The \emph{distortion} of an embedding is the product of its contraction and expansion $\lambda = \lambda_c \cdot \lambda_e$.
\end{definition}

Note that this definition generalizes \cref{def: embedding}. Specifically, for non-contractive embeddings, the definitions coincide and produce the same notion of distortion. It is also easy to verify that \cref{thm: reduction} can be adapted to incorporate this more general notion of distortion by multiplying the factor $\lambda_c$ in the corresponding step in \cref{claim: step2}.

One might notice some asymmetry between contraction and expansion for probabilistic embeddings. The reason behind this difference is as follows: contraction bounds the loss when mapping the cost incurred by an online algorithm in the embedded (target) space back to the original space $(X, d_X)$. Intuitively, the online algorithm prefers to move mostly along highly contracted distances, because these are cheap in the target space. Thus, we require the contraction to be bounded in \emph{every} realization. In contrast, expansion bounds the loss when translating an optimal offline solution in $(X, d_X)$ into the embedded target space. Because the optimal offline solution is independent of the random choices of the online embedding, the maximum expansion can be safely evaluated outside the expectation. 

\section{Handling Unknown Length}
\label{appendix: unknown n}

Previous online embedding algorithms \cite{Bartal2020, Indyk10online} manage unknown sequence lengths by dynamically adjusting the component construction. Unfortunately, our approach does not directly support this method: as more points arrive, we must decrease the allowed merge-failure probability $\gamma$ to offset the increase in relevant scales $O(n)$. However, the probability of merge failure (caused by another merge) monotonically increases with the number of points. Thus, dynamically selecting $\varepsilon$ alone is insufficient for our embedding method.

We could adapt our embedding if we allow it to have a limited amount of contraction (i.e., sometimes underestimate distances). See \cref{def: embedding2} for the definition of contraction and distortion for potentially contractive embeddings. We now demonstrate that without knowing $n$, it is possible to achieve an embedding with distortion $O(\log^2 n \log\log n)$, where the contraction is $O(\log\log n)$. Intuitively, we permit up to $O(\log\log t)$ components to merge at any time $t$.

We need to modify \cref{lemma: hierarchy} slightly.

\begin{restatable}[HST Construction]{lemma}{2}
\label{lemma: hierarchy 2}
    Suppose that for every integer $j$, there is a probabilistic online $(\alpha \cdot 2^j)$-bounded monotone partition on $(X,d_X)$ for an update sequence $\sigma$ of length $n$ and that for each time $t \le n$, the probabilistic partition $C_t$ is $(\delta_t^{(j)}, \varepsilon, \gamma)$-smooth. Then there is an online monotone embedding from $(X,d_X)$ into HSTs that, on the same input, achieves a distortion of 
    \[
    O\left(\alpha \cdot \max_{t \le n}\left(\max_{j \in \mathbb{Z}} \delta_t^{(j)} \log \varepsilon^{-1} + \gamma \sum_{j \in \mathbb{Z}} \delta_t^{(j)}\right)\right).
    \]
\end{restatable}
\begin{proof}[Proof Sketch]
The expansion part is bounded the same as in \cref{lemma: hierarchy}. To verify the contraction bound, note that any pair \( u,w \) will always be separated at any scale \( s \le d(u,w)/\alpha\).
\end{proof}

\begin{theorem}
    There exists a probabilistic online monotone embedding of up to $n$ points from any metric space into HSTs with distortion $O(\log^2 n \log \log n)$ for any $n \in \mathbb{N}$, without prior knowledge of $n$.
\end{theorem}

\begin{proof}
Fix a scale \( s \). We outline the construction of a partition that is \((2\log\log t)\cdot s\)-bounded and \((O(\log t), t^{-12}, O(t^{-1}))\)-smooth for all times \( t \). The theorem then follows from \cref{lemma: hierarchy 2}.

The algorithm operates in phases \( i = 1, 2, \dots \). During each phase \( i \), the algorithm maintains a guess \( m_i = 2^{2^i} \) for the total number of points. Phase \( i \) begins when the \((m_{i-1}+1)\)th point arrives and ends immediately before the \((m_i+1)\)th point arrives. Thus, at any time \( t \) during phase \( i \), we have \( t \in [\sqrt{m_i}, m_i] \).

The base partition \( C_t \) is constructed as in the known-\( n \) case, ensuring \( \delta = O(\log t) \). However, rather than allowing only one designated merging component globally, we now permit exactly one designated merging component \emph{per phase}. Specifically, in phase \( i \), we set \(\varepsilon_i = m_i^{-6}\). During phase \( i \), a component attempts to merge if it separates two points within distance \(\varepsilon_i s\). The first component to attempt a merger within each phase is permitted to merge freely, while all subsequent merger attempts involving different components are rejected. Using a similar analysis as in \cref{thm: incremental partition}, during phase \( i \), the probability of separating a close pair \( u,w \) with \( d(u,w) \le \varepsilon_i s \) is bounded by 
\[
O(\varepsilon_i m_i^{5}\log m_i)\frac{d(u,w)}{s}.
\] 
Hence $\gamma = O(\varepsilon_i m_i^{5})$. By setting \(\varepsilon_i = m_i^{-6} \ge t^{-12}\), we achieve \(\gamma = O(m_i^{-1}) \le O(t^{-1})\), which suffices to ensure the desired smoothness conditions. 

Lastly, to verify the partition is indeed \((2\log\log t)\cdot s\)-bounded, note that at most one component per phase can successfully merge. Thus, by time \( t \), there are at most \( \log\log t \) designated merging components, resulting in clusters with diameters of at most \( 2(\log\log t)\cdot s \).
\end{proof}

\section{Deterministic Embeddings}
\label{appendix: deterministic parent}
We now discuss the results for deterministic embeddings in the incremental setting, summarised in Table A. 

\begin{table}[ht]
    \centering
    \label{tab: Det}
    \begin{tabular}{| c| c| }
    \hline
        Setting & Distortion\\
        \hline\hline
        Offline & $n-1$ (\cref{thm: Offline})\\
        Online strict & $2^n$ (\cref{thm: online_lower}, \ref{thm: online_upper})\\
        Online monotone, known $n$ & $n-1$ (\cref{thm: deterministic})\\
        Online monotone, unknown $n$ & $\widetilde \Theta (n\log n)$ (\cref{thm: deterministic})\tablefootnote{The $\tilde\Theta$ hides poly-$\log \log n$ factors here, see \cref{rem: don't know n} for detail.} \\
        \hline
    \end{tabular}
    \caption{Deterministic embeddings into HSTs}
\end{table}

\subsection{Offline and Strict Online Embedding into HSTs}
Offline and strict online embedding into trees exhibit a distortion of $\Theta(n)$ (\cite{Rabinovich1998LowerBO, Gupta01Steiner}) and $2^{(n-4)/2}$ respectively. For HSTs, the bounds are the same asymptotically, but more specific. We believe this bound is known, but we have not found any literature that provides proof, so we present our proof here for completeness.

\begin{theorem}[Folklore]
\label{thm: Offline}
    There is an (offline) embedding of $n$ points into HSTs with distortion $n-1$. This bound is tight.
\end{theorem}

\begin{proof}
    We start with the lower bound. Consider a line of length $n-1$ with $n$ points evenly distributed at a distance of $1$ from each other. Denote the points from left to right as $v_1, \ldots, v_n$, each with coordinate $v_i = i-1$. Assume without loss of generality that the embedding is non-contractive (otherwise, we multiply all distances by the contraction factor), then $d_T(v_1, v_n) = n-1$. Since HST is an ultrametric, for all $u, v, w$, $d_T(u, w) \leq \max(d_T(u, v), d_T(v, w))$. Thus, inductively $d_T(v_1, v_k) \leq \max_{i=1}^{k-1} \{d_T(v_t, v_{i+1})\}$. Therefore, $\max_{i=1}^{k} d_T(v_t, v_{i+1}) \geq d_T(v_1, v_n) = n-1$. Suppose $d_T(v_j, v_{j+1}) \geq n-1$ for some $j$, given $d(v_j, v_{j+1}) = 1$, the distortion is at least $n-1$.

    We now describe the construction that achieves $n-1$ distortion. The approach is inspired by Kruskal's algorithm \cite{kruskal} for minimum spanning trees. We start with $n$ singleton leaves $v$, each with $\varphi(v)=0$. We sort all pairwise distances in increasing order and process them one by one. For a pair of points $u, v$ with distance $d$, if they are already in the same HST, do nothing. Otherwise, we find the roots of $u$ and $v$, denoted by $r_u$ and $r_v$. Add a node $r$ with value $\varphi(r) = \varphi(r_u) + \varphi(r_v) + d$, making $r_u$ and $r_v$ its children.

    By induction, at any stage in an HST $T$ with $k$ points, for all $u, v$ within $T$, $d_T(u, v) \leq (k-1)d(u, v)$. The base case for singleton trees is trivial. When combining two trees $T_1, T_2$ with sizes $k_1, k_2$ and roots $r_1, r_2$ using points $u, v$ where $d(u, v) = d$, for any $u' \in T_1, v' \in T_2$, we must have $d(u', v') \geq d$, otherwise $u', v'$ would have merged earlier. Additionally, the edges added earlier must have values $\leq d$, hence $\varphi(r_1) \leq (k_1-1)d$ and $\varphi(r_2) \leq (k_2-1)d$. Thus:
    \[\varphi(r) = \varphi(r_1) + \varphi(r_2) + d \leq (k_1 + k_2 - 1)d \leq (k_1 + k_2 -1)d(u', v').\]
    At the end of the process, all points are combined into a single HST with a distortion of $n-1$. 
\end{proof}

\begin{theorem}
    There is a non-contractive online embedding of $n$ points into HSTs with distortion $2^n$, even without prior knowledge of $n$.
\label{thm: online_upper}
\end{theorem}
\begin{proof}
    We present a construction. Upon the arrival of the $i$th point $x$, find the nearest point $u$ in the original metric space. Attach $x$ to the ancestor of $u$ with the smallest value such that for all $v$, $d_T(x, v) \geq \left(1+2^{-(i-1)}\right) d(x, v)$. 
    
    We prove by induction that after the arrival of the $i$th point, for every pair $u, v$, we have $d_T(u, v)\in [(1+2^{-(i-1)})d(u, v), 2^i d(u, v)]$. For $i=1,2$, the claim is trivial. Assuming it holds for $i > 2$, the lower bound is straightforward by construction. Now we consider the upper bound.

    Let $S = \br {v \mid  d_T(u, v) \leq \left(1+2^{-(i-1)}\right) d(x, v)}$, then $d_T(x, u) = \max_{v \in S}\br {\left(1+2^{-(i-1)}\right) d(x, v)}$.
    
    We now establish a bound for $d_T(x, u)$. For every $v \in S$,
    \begin{align*}
        \left(1+2^{-(i-2)}\right)  d(u,v) \leq d_T(u, v) &\leq \left(1+2^{-(i-1)}\right)  d(x, v) \\
        &\leq \left(1+2^{-(i-1)}\right) (d(x, u) + d(u, v)),
    \end{align*}
    we deduce that $\left(1+2^{-(i-1)}\right)  d(x, u) \geq 2^{-(i-1)} d(u, v)$. Thus
    \[d(x, v) \leq d(x, u) + d(u,v) \leq \left( 1 + 2^{i-1}\left(1+2^{-(i-1)}\right) \right) d(x, u)=(2^{i-1}+2)d(x, u).\]
    Therefore, we conclude
    \begin{align*}
        d_T(x, u) = \left(1+2^{-(i-1)}\right)  d(x, v) &\leq \left(1+2^{-(i-1)}\right) \left( 2^{i-1} + 2\right) d(x, u) \leq \left(2^{i-1} + 4\right) d(x, u)
    \end{align*}
    Given that $i \geq 3$, we conclude that $d_T(x, u) \leq 2^i d(x, u)$.
    
    For $v \neq u$, we have $d_T(x, v) = \max(d_T(x, u), d_T(u, v))$.
    Since $d_T(x, u) \leq 2^t d(x, u) \leq 2^i d(x, v)$
    and 
    \begin{align*}
        d_T(u, v) \leq 2^{i-1} d(u, v) &\leq 2^{i-1} (d(x, u) + d(x, v)) \\
        &\leq 2^{i-1}(d(x, v) + d(x, v)) = 2^i d(x, v),
    \end{align*}
    we conclude $d_T(x, v) \leq 2^i d(x, v)$, which completes our proof.
\end{proof}

\subsection{Deterministic Online Monotone Embedding}
\label{appendix: deterministic}
We first define some properties of deterministic partitions and relate them to HST embedding. Recall the definition of $s$-bounded to compare with.

\begin{definition}
    A partition $C = \br {C_j}$ of points $V$ on a metric $d$ is $t$-connected if for all $u, v$ such that $d(u, v) \leq t$, $C(u) = C(v)$.
\end{definition}

In other words, $s$-bounded says that the points in the same cluster have distances of at most $s$ from each other, and $t$-connected says that the points in different clusters have distances of more than $t$ from each other.
    
\begin{lemma}
\label{thm: rela - det}
    Consider a set of scales $\br{s_j}$, partitions $\br{C^j}$ for every $j \in \mathbb Z$, and the induced HST embedding:
    \begin{enumerate}
        \item if for all $s$, the partition $C$ at scale $s$ is $\alpha\cdot s$-bounded for some constant $\alpha \geq 1$, then the contraction of the embedding is at most $\alpha$. In particular, the HST embedding is non-contractive if and only if for all scale $s$, the partition is $s$-bounded.
        \item if for all $s$, the partition $C$ at scale $s$ is $s / \beta$-connected for some constant $\beta \geq 1$, and there exists $j$ such that $s_j = \beta d(u, v)$, then the expansion of the constructed embedding is at most $\beta$.
    \end{enumerate}
\end{lemma}
\begin{proof}
    For two points $u, v \in V$, the tree distance $d_T(u, v)$ is at least the smallest scale $s$ in which $C(u) = C(v)$. Suppose this scale is $s_j$. If the partitions for all scales are $\alpha \cdot s$-bounded, then $ d(u, v) \le \alpha s_j \le \alpha d_T(u, v)$. Hence, the contraction is at most $\alpha$.
    
    On the other hand, suppose the partition at each scale $s$ is $s / \beta$-connected and the scale $\beta d(u, v)$ is present, then at scale $s_j = \beta d(u, v)$ or above, the partition must have $C(u) = C(v)$. Thus $d_T(u, v) \leq s_j \leq \beta d(u, v)$. Note that it is necessary that $\beta d(u, v)$ is one of the scales. Otherwise, $d_T(u, v)$ can only be upper bounded by the first scale $s_j \geq \beta d(u, v)$.
\end{proof}

We are ready to prove \cref{thm: deterministic}. We first prove the non-expansive monotone embedding with distortion $n-1$. 
\Deterministic*
\begin{proof}[Proof with known $n$]
    We prove a stronger result, a non-expansive embedding with contraction $n-1$ without the knowledge of $n$. Clearly, one can turn this into non-contractive by scaling all the distances by $n-1$.
    
    Since we have no knowledge of $n$, we maintain an embedding of the first $i$ points $V_t$ with distortion $i-1$ by maintaining a set of scales and a set of partitions of points $V_t$ for each scale. The set of scales will be $\br {d(u, v) \mid u, v\in V_t}$, the set of pair-wise distances. For each scale $s_j$, we maintain an $(i-1)s_j$-bounded and $s_j$-connected partition. Recall that this means for each cluster $S \in C^j$ and $u, v \in S$, $d(u, v) \in [s_j, (i-1)s_j]$. By \cref{thm: rela - det}, this guarantees a non-expansive embedding with contraction $i-1$. This bound is tight as $n-1$ distortion matches the offline optimum.

    The algorithm proceeds by iterating over all pairwise distances $d(u,v)$ in decreasing order after each point's arrival. For each scale $s_j$, we construct a graph linking points $u, v$ if $d(u, v) \leq s_j$. The partition at this scale is induced by the connected components.
    It is clear that if two points are in the same cluster at time $i-1$, they will remain in the same cluster at time $i$, as we only add more edges when a new point arrives. Thus, the partition satisfies the monotonicity requirement.

    This bound is tight since it matches the offline lower bound.
\end{proof} 

Interestingly, without knowledge of $n$, directly constructing a non-contractive online monotone embedding turns out to have a more complex distortion bound.
\begin{proof}[Proof for the non-contractive case]
    We first show necessity. Consider $n$ points aligned on a line. Let the first point be $v_1$ at coordinate $0$. Introduce a new point $v_2$ at coordinate $1/f(2)$. Since the algorithm lacks foreknowledge of $n$, the distortion must not exceed $f(2)$ instead of $f(n)$, implying $d_T(v_1, v_2) \leq d(v_1, v_2) \cdot f(2) = 1$, and it cannot increase subsequently. Similarly, add $v_3$ at coordinate $1/f(2) + 1/f(3)$, ensuring $d_T(v_2, v_3) \leq 1$. This pattern continues for $v_t$ at coordinate $\sum_{j=2}^t 1/f(j)$, and we have $d_T(v_{t-1}, v_t) \leq 1$. Since the HST is an ultrametric, we deduce:
    \[
    d_T(v_1, v_n) \leq \max_{i=1}^{n-1} \{d_T(v_t, v_{i+1})\} \leq 1,
    \]
    yet $d(v_1, v_n) = \sum_{i=2}^{n} 1/f(i)$. Therefore, for any $f$ whose reciprocal sum diverges, a non-contractive guarantee is unfeasible.

    Next, we demonstrate an algorithm achieving distortion $O(f(n))$ for $f$ whose reciprocal sum converges. Specifically, the distortion will be $f(n) / (\sum_{i=2}^\infty 1/f(i))$. Essentially, we can attain exactly $f(n)$ if $f$ is non-decreasing and $\sum_{n=2}^\infty 1/f(n) \leq 1$. Assuming such conditions, after the $i$th point arrives, we maintain a set of scales $S = \br {d(u, v) \cdot f(k) \mid u, v \in V_t, 1 \leq k \leq i}$. For each scale $s$, we maintain a partition that is $s$-bounded and $s / f(i)$-connected.
    
    At scale $s$, we construct the graph such that an edge exists between $v_j, v_k$ if $j < k$ and $d(v_j, v_k) \leq s/ f(k)$. This process can be visualized as each new point $v_k$ linking to all preceding points $v_j$ if the distance is at most $s/f(k)$. The partition is again the connected components in the graph. Since we only insert edges over time, the partition satisfies the monotone constraints.
    
    The partition is clearly $s / f(i)$-connected from the construction since $f(i) \geq f(j)$ for $i > j$, and scale $s = f(i)\cdot d(u, v)$ is present. To prove $s$-boundedness, we note that the longest chain possible is $\sum_{i=2}^t s/f(i) \leq s$. By Theorem \ref{thm: rela - det}, this partition implies a non-contractive HST embedding with distortion $O(f(n))$.
\end{proof}

\begin{remark}
\label{rem: don't know n}
    The bound for non-contractive embeddings is not straightforward to articulate, as there is not a largest converging series \cite{Ash97}. However, a characterization can be given. Let $\log^{(k)} n$ represent taking $k$ times the logarithm of $n$, ensuring each result is at least $1$, i.e.,
    \[
    \begin{aligned}
    \log^{(0)} n &= n,\\
    \log^{(k+1)} n &= \max(\log (\log^{(k)}n), 1).
    \end{aligned}
    \]
    
    Let $k \geq 0$, the reciprocal sum $\sum_{i=2}^\infty 1/f(i)$ diverges for \[f_k(n) = \prod_{i=0}^k \log^{(i)} n\] but converges for \[g_k(n) = f_k(n) \cdot \log^{(k)} n.\] 
    For $k=1$, this yields a simplified result of $\Omega(n\log n)$ and $O(n\log^2 n)$.
\end{remark}

    We find the contrast between non-contractive and non-expansive embeddings quite interesting, and we provide some intuitive explanations here.

    First, if we require the embedding to be non-contractive, then knowing $n$ has $O(n)$ distortion while not knowing $n$ has $\Omega(n\log n)$. The reason is that for any fixed $n$, at any time step $t < n$, if $n$ is known, the distortion could be $f(n)$, but otherwise the distortion has to be $f(t)$. Thus, knowing $n$ gives a more relaxed constraint.
    
    Second, if the algorithm does not know $n$, then non-contractive embedding has distortion $\Omega(n \log n)$ but non-expansive embedding has distortion $O(n)$. This is due to our constraint that distances do not increase. Consider a pair of points $u, v$. If the embedding is non-expansive, the valid range for $d_t(u, v)$ after the $t$th point arrives and after $(i+1)$th point arrives is $[d(u,v)/(t-1), d(u,v)]$ and $[d(u,v)/t, d(u,v)]$, respectively. Note that this change is in the same direction as what we allow, decreasing the distances. Thus, if the algorithm finds a benefit in having a smaller $d_T(u, v)$ after more points arrive, it could do so. However, if the embedding is non-contractive, the range for $d_t(u, v)$ is $[d, (t-1)d]$ and $[d, td]$. Unfortunately, we do not allow an increase in the distance, and thus the algorithm cannot make use of the new degree of freedom. This explains the different behaviors between non-contractive and non-expansive embeddings.

\section{Application}
\label{appendix: application}
Here we include the full proof of \cref{thm: reduction}.
\begin{proof}
    Recall the definition of $\seq M$, $\Alg^{\seq M}$ and $\Alg$. Let $\seq b\in A^m$ be the answer sequence of the optimal offline algorithm for the original game $G_N$. Let $\gamma(\seq M) \in \D(A^m)$ be the distribution of the answers given by $\Alg^{\seq M}$. We use $\E_{\seq M}$ to denote expectation over the randomness of the embedding and use $\E_{\seq a \sim \gamma(\seq M)}$ for expectation over the randomness of $\Alg^{\seq M}$.

    \begin{claim}
    \label{claim: step1}
        For each $\seq M$,
    \begin{align*}
        \E_{\seq a\sim \gamma(\seq M)}\left[\sum_{t=1}^m c^{M_t}_t(\seq r_{[0, t]}, \seq a_{[0, t]}) \right] \leq \rho \sum_{t=1}^m c^{M_t}_t(\seq r_{[0, t]}, \seq b_{[0, t]})+\Phi_0.
    \end{align*}
    \end{claim}
    \begin{proof}
        Let $\gamma_t \in \D(A^t)$ be the distribution of the first $i$ answers according to $\gamma(\seq M)$. For every $i$,
        \begin{align*}
            \E_{\seq a_{[0, t]}\sim\gamma_{t}}\left[c_{t}^{M_{t}}(\seq r_{[0, t]}, \seq a_{[0, t]})\right] &\le \rho c_{t}^{M_t}(\seq r_{[0, t]}, \seq b_{[0, t]})-\Phi_{t}^{M_{t}}(\seq r_{[0, t]}, \gamma_{t}, \seq b_{[0, t]})+\Phi_{t-1}^{M_{t}}(\seq r_{[0, t-1]},\gamma_{t-1},\seq b_{[0, t-1]})\\
            &\le \rho c_{t}^{M_t}(\seq r_{[0, t]}, \seq b_{[0, t]})-\Phi_{t}^{M_{t}}(\seq r_{[0, t]}, \gamma_{t}, \seq b_{[0, t]})+\Phi_{t-1}^{M_{t-1}}(\seq r_{[0, t-1]},\gamma_{t-1},\seq b_{[0, t-1]}),
        \end{align*}
        where the first inequality is by definition of $\Alg^{\seq M}$ and the second inequality is by monotonicity of the potential function. Summing up for each $1\le i\le m$ achieves the claim.
    \end{proof}
    
    Since the embedding is non-contractive, the cost of $\Alg^{\seq M}$ in $G_N$ is upper bounded by its cost in $G_{\seq M}$:
    \begin{claim}
    \label{claim: step2}
    For every $\seq M$ and $i$, $\E_{\seq a\sim \gamma(\seq M)} [c_t^{N}(\seq r_{[0, t]}, \seq a_{[0, t]})] \leq \E_{\seq a\sim\gamma(\seq M)}[c_t^{M_t}(\seq r_{[0, t]}, \seq a_{[0, t]})]$.
    \end{claim}
    
    Finally, we bound the offline cost in $G_{\seq M}$ by its cost in $G_N$.
    \begin{claim}
    \label{claim: step3}
        For every $i$, $\E_{\seq M}[c_t^{M_t}(\seq r_{[0, t]}, \seq b_{[0, t]})] \leq \lambda c_t^{N}(\seq r_{[0, t]}, \seq b_{[0, t]})$.
    \end{claim}
    \begin{proof}
    Let us expand $c_t^{M_t}$:
        \begin{align*}
            \E_{\seq M}[c_t^{M_t}(\seq r_{[0, t]}, \seq b_{[0, t]})] &= \E_{\seq M}\left [
    \sum_{u, v} \alpha_t(u, v, \seq r_{[0, t]}, \seq b_{[0, t]}) d_t(u, v) + \beta_t(\seq r_{[0, t]}, \seq b_{[0, t]})  \right]\\
    &= \sum_{u, v} \alpha_t(u, v, \seq r_{[0, t]}, \seq b_{[0, t]}) \E_{\seq M}[d_t(u, v) ] + \beta_t(\seq r_{[0, t]}, \seq b_{[0, t]})\\
    &\le \sum_{u, v} \alpha_t(u, v, \seq r_{[0, t]}, \seq b_{[0, t]}) d_N(u, v) \lambda + \beta_t(\seq r_{[0, t]}, \seq b_{[0, t]}) \lambda\\
    &=\lambda c_t^{N}(\seq r_{[0, t]}, \seq b_{[0, t]}). &\qedhere
        \end{align*}
    \end{proof}
    Combining Claims~\ref{claim: step1}, \ref{claim: step2} and \ref{claim: step3} yields 
    \begin{align*}
        \mathrm{cost}_{\Alg}(\seq r) &=  \E_{\seq M}\left[\E_{\seq a\sim \gamma(\seq M)} \left[\sum_{t=1}^m c_t^N(\seq r_{[0, t]}, \seq a_{[0, t]})\right]\right] \leq \E_{\seq M}\left[\E_{\seq a\sim \gamma(\seq M)}\left[\sum_{t=1}^m c_t^{M_t}(\seq r_{[0, t]}, \seq a_{[0, t]}) \right]\right]\\
        &\leq \E_{\seq M} \left[ \left(\rho \sum_{t=1}^m c_t^{M_t}(\seq r_{[0, t]}, \seq b_{[0, t]}) + \Phi_0\right)\right]\leq \rho \lambda \sum_{t=1}^m c_t^{N}(\seq r, \seq b) + \lambda\Phi_0.
    \end{align*}
    Taking $\eta = \lambda \Phi_0$, we conclude $\mathrm{cost}_{\Alg}(\seq r) \leq \rho \lambda \cdot \mathrm{opt}(\seq r) + \eta$, showing $\Alg$ is $\rho \lambda$-competitive.
\end{proof}

\subsection{\texorpdfstring{The $k$-Server Reduction}{The k-Server Reduction}}
\label{appendix: kserver}
We now prove \cref{thm: k-server}. We employ the results from \cite{Bubeck18}, which provide an \( O(\log^2 k) \)-competitive fractional \( k \)-server algorithm on HSTs and use the online rounding technique from \cite{bansal15polylog} to obtain a randomized online algorithm with a constant factor loss. We first show that their potential function is monotone, and then verify that online rounding can still be performed when the branches of the HST are merged. Both claims are proven in \cite{Lee2018FusibleHA}, which was later withdrawn due to bugs in other sections.

\paragraph{Notations in \cite{Bubeck18}.} The paper first presents a fractional $k$-server algorithm, which is allowed to assign non-integral numbers of servers, and must move a total mass of $1$ to the request location for service. The algorithm maintains a fractional \( k \)-server configuration using the following notations.

We use $V$ to denote the set of vertices, $L$ to denote the set of leaves, $ch_u$ to denote the set of children of $u$, and $N_u$ to denote the set of leaves in the subtree of $u$. For \( u \in V \), define
\begin{align*}
\chi(u) &= \left\{ (v,j) : v\in ch_u, j \in [N_v] \right\},\\
\Lambda &= \left\{ (r,i) : i \in [N_r] \right\} \cup \bigcup_{u \in V} \chi(u).
\end{align*}
With a slight abuse of notation, we sometimes write \( \sum_{i \ge 1} f(x_{u,i}) \) instead of \( \sum_{i \in [N_u]} f(x_{u,i}) \).

The \emph{assignment polytope} on an HST is defined as
\[
\mathcal{A} := 
\left\{
x \in [0,1]^{\Lambda} :
\sum_{i \le |S|} x_{u,i} \le \sum_{(v,j) \in S} x_{v,j} \quad
\forall\, u \in V,\, S \subseteq \chi(u),
\quad
x_{r,i} = \mathbbm{1}_{\{i > k\}} \quad \forall i \ge 1
\right\}.
\]

Finally, for \( u \in V \) and \( x \in \mathcal{A} \), define the \emph{associated server measure}
\( z \in \mathbb{R}_+^{V} \) by
\[
z_u := \frac{1}{1 - \delta} \sum_{i \ge 1} (1 - x_{u,i}).
\]

The paper maintains the $x$ variables using a mirror descent flow. Intuitively, \( x \) represents the fractional anti-server mass. The $z$ variables on the leaves yield a fractional algorithm with \( \frac{k}{1-\delta} = k + \varepsilon \) servers, which the paper shows could be turned into a fractional solution with $k$ servers if $\delta = 1/(3k)$ and $\varepsilon = 1/3$.

To define the potential function, we associate a similar (anti-server) measure $y$ to the offline configuration, where $y_{u, i} = 1$ iff there are less than $i$ offline servers in the subtree of $u$. The potential function employed is \(\Phi(y, x) = C_0 D(y, x) - H(x)\), where $C_0 = \Theta(\log k)$ and
\begin{align*}
    D(y, x) &= \sum_{u \in V \setminus \{r\}} w_u \sum_{i \ge 1} (y_{u,i} + \delta) \log \left( \frac{y_{u,i} + \delta}{x_{u,i} + \delta} \right),\\
    H(x) &=
    \sum_{u \in V \setminus \{r\}} 
    w_u \left[
    \left( z_u + \left( 1 + \tau^{-1} \mathbbm{1}_{\{u \notin \mathcal{L}\}} \right) \varepsilon \right)
    \log\left( z_u + \varepsilon \right)
    + z_u \log\left( z_{p(u)} + \varepsilon \right)
    \right].
\end{align*}

\paragraph{Maintaining the measure under merges.} 
Note that the potential is defined in terms of the measure $x$, so we need to explain how to maintain them under a merge operation. We will consider merging two sibling vertices $u$, $u'$, and call the merged vertex $u^*$. Whenever two vertices \( u, u' \) are merged, merge the respective lists \( \{x_{u,i}\} \) and \( \{x_{u',i}\} \), and sort the combined list in non-decreasing order. Consequently, \(z_{u^*}= z_u + z_{u'} \). It should be easy to verify that the resulting assignment lies within the assignment polytope of the new HST. 

\paragraph{Monotonicity of the potential.} The Bregman divergence term \( D(y, x) \) decreases because we sort the terms, thereby narrowing the distance between \( x \) and \( y \). For \( H(x) \), ignoring the \( \varepsilon \) term gives the simpler form \( \sum w_u z_u \log z_u \). When merging \( u \) and \( u' \), the term \( z_u \log (z_{p(u)} + \varepsilon) + z_{u'} \log (z_{p(u)} + \varepsilon) \) remains unchanged, while
\[
z_u \log z_u + z_{u'} \log z_{u'} \le (z_u + z_{u'}) \log (z_u + z_{u'})
\]
follows from the convexity of \( f(t) = t \log t \). Since the potential involves \( -H(x) \), merging decreases the potential.

We now formalize the above intuitive sketch.

\begin{claim}
    The term \( D(y, x) \) decreases under a merge step.
\end{claim}

\begin{proof}
     Consider function \( f(a, b) = b \log (b/a) \). Then \( f \) is submodular since
    \[
    \frac{\partial^2 f}{\partial x \partial y} = -\frac{1}{a} < 0.
    \]
    Hence, by Topkis's Theorem, the sum \( \sum f(a_i, b_{\pi(i)}) \) is minimized when the sequences are paired in the same order—e.g., both sorted\footnote{Topkis's Theorem says that supermodular functions are maximized when both sequences are sorted, hence submodular functions are minimized in these cases.}. Consider \( u \) and \( u' \). Let \( x_{u^*} \) be the merged and sorted concatenation of \( x_u \) and \( x_{u'} \). Similarly, let \( y_{u^*} \) be the corresponding offline vector after merging, also sorted. We have
    \[
    \sum_i (y_{u,i} + \delta) \log \left( \frac{y_{u,i} + \delta}{x_{u,i} + \delta} \right)
    +  \sum_i (y_{u',i} + \delta) \log \left( \frac{y_{u',i} + \delta}{x_{u',i} + \delta} \right)
    \ge
    \sum_i (y_{u^*,i} + \delta) \log \left( \frac{y_{u^*,i} + \delta}{x_{u^*,i} + \delta} \right). \qedhere
    \]
\end{proof}

\begin{claim}
    The term \( -H(x) \) decreases under a merge step.
\end{claim}
\begin{proof}
    Note that the second term \( z_u \log\left( z_{p(u)} + \varepsilon \right) + z_{u'} \log\left( z_{p(u)} + \varepsilon \right)\) remains unchanged. For the first term, we show that convexity still holds when considering the \( \varepsilon \)-term. To do so, we aim to show that for any numbers \( a, b, c \ge 0 \) and \( 0 \le \varepsilon \le 1 \) such that \( (1 + c)\varepsilon \le 1 \), it holds that
\[
\bigl(a + b + (1 + c)\varepsilon\bigr) \log(a + b + \varepsilon)
\ge 
\bigl(a + (1 + c)\varepsilon\bigr) \log(a + \varepsilon)
+ 
\bigl(b + (1 + c)\varepsilon\bigr) \log(b + \varepsilon).
\]
    Assuming this is true, substitute $a = z_u, b = z_{u'}, c=\tau^{-1}\mathbbm 1_{u \notin L}$ shows that $H(x)$ increases and hence $-H(x)$ decreases. 
    
    We are now left with proving the inequality. Define, for \( t \ge 0 \),
\[
f(t)
=\bigl(a+t+(1+c)\varepsilon\bigr)\log(a+t+\varepsilon)
-\bigl(a+(1+c)\varepsilon\bigr)\log(a+\varepsilon)
-\bigl(t+(1+c)\varepsilon\bigr)\log(t+\varepsilon).
\]
Then \( f(0)=0 \), and it suffices to show \( f'(t)\ge 0 \) for all \( t\ge 0 \). Differentiate:
\[
\begin{aligned}
f'(t)
&= \log(a+t+\varepsilon) + \frac{a+t+(1+c)\varepsilon}{a+t+\varepsilon}
 - \log(t+\varepsilon) - \frac{t+(1+c)\varepsilon}{t+\varepsilon} \\
&= \log\Bigl(1+\frac{a}{t+\varepsilon}\Bigr)
  + (1+c)\varepsilon\left(\frac{1}{a+t+\varepsilon}-\frac{1}{t+\varepsilon}\right) \\
&= \log\Bigl(1+\frac{a}{t+\varepsilon}\Bigr)
  - (1+c)\varepsilon\cdot \frac{a}{(t+\varepsilon)(a+t+\varepsilon)}.
\end{aligned}
\]
Let \( s=\dfrac{a}{t+\varepsilon}\ge 0 \). Then
\[
\frac{a}{(t+\varepsilon)(a+t+\varepsilon)}=\frac{s}{(1+s)(t+\varepsilon)},
\]
so
\[
f'(t)
= \log(1+s)\;-\;\theta\,\frac{s}{1+s},
\qquad
\text{where }\ \theta=\frac{(1+c)\varepsilon}{t+\varepsilon}.
\]
Since \((1+c)\varepsilon \le 1\) and \(t+\varepsilon \ge \varepsilon\), we have \(0\le \theta\le 1\). Noting that \(\log(1+s) \ge \frac{s}{1+s} \ge \theta \frac{s}{1+s}\), we obtain \(f'(t) \ge 0\), as desired.
\end{proof}

\begin{corollary}
\label{corr: potential server}
    The potential function $\Phi(y, x)$ decreases under a merge step.
\end{corollary}

\paragraph{Rounding under merges.} Theorem 5.2 of \cite{bansal15polylog} gives an online rounding of a fractional k-server algorithm on HSTs to a distribution over integral algorithms that only loses a constant factor in the expected cost. Here we show that the results preserve in the presence of merging.

\begin{lemma}[Theorem 5.16 \cite{Lee2018FusibleHA}]
\label{lem:rounding-merge}
An online fractional $k$-server algorithm on HST metrics that may evolve via merges can be transformed into an online randomized integral algorithm, incurring only an $O(1)$ multiplicative increase in cost.
\end{lemma}
\begin{proof}
    The key property maintained by the rounding in \cite{bansal15polylog} is that the integral algorithm is supported on balanced configurations with respect to the fractional algorithm, i.e., let $z, z'$ be the measures of the fractional and random integral configuration, then at all times $t$ and all supports of the random algorithm,
    \[z'_u \in \{\lfloor z_u \rfloor, \lceil z_u \rceil \}\]
    
    We describe how to maintain the balance property when two sibling vertices $u, u'$ are merged into $u^*$. Suppose that $z'$ is a random integral $k$-server measure satisfying, for two siblings $u, v$ in the tree,
    \[
    \mathbb{E}[z'(u)] = z(u), \qquad
    \mathbb{E}[z'(v)] = z(v),
    \]
and with probability one, $z'$ satisfies the balance conditions:
\[
 z'(u) \in [\lfloor  z(u) \rfloor,\, \lceil  z(u) \rceil],
\quad
 z'(v) \in [\lfloor  z(v) \rfloor,\, \lceil  z(v) \rceil].
\]

For simplicity, let us denote $\varepsilon_u = z_u - \lfloor z_u \rfloor, \varepsilon_v = z_v - \lfloor z_v \rfloor$. We produce a random variable $(k_u, k_v)$ with the following properties:
\begin{enumerate}
  \item $\mathrm{supp}((k_u,k_v)) \subseteq \mathrm{supp}((z'_u, z'_v))$.
  \item $\mathbb{P}(k_u = \lfloor z_u \rfloor) = \mathbb{P}(z_u' = \lfloor z_u \rfloor)$.
  \item $\mathbb{P}(k_v = \lfloor z_v \rfloor) = \mathbb{P}(z'_v = \lfloor z_v \rfloor)$.
  \item The balance condition is satisfied:
  \[
  \mathbb{P}\big(k_u + k_v \in \{\lfloor z_u + z_v \rfloor, \lceil z_u + z_v \rceil\}\big) = 1.
  \]
\end{enumerate}
There are two cases to handle, depending on the fractional parts $\varepsilon_u$ and $\varepsilon_v$.

\paragraph{Case 1: $\varepsilon_u + \varepsilon_v \le 1$.}
\[
\begin{aligned}
\mathbb{P}[(k_u,k_v) = (\lfloor z_u \rfloor, \lfloor z_v \rfloor)] &= \mathbb{P}( z'_u = \lfloor z_u \rfloor) + \mathbb{P}( z'_v = \lfloor z_v \rfloor) - 1,\\
\mathbb{P}[(k_u,k_v) = (\lfloor z_u \rfloor, \lceil z_v \rceil)] &= \mathbb{P}( z'_v = \lceil z_v \rceil)\mathbbm{1}_{\{\varepsilon_v>0\}},\\
\mathbb{P}[(k_u,k_v) = (\lceil z_u \rceil, \lfloor z_v \rfloor)] &= \mathbb{P}( z'_u = \lceil z_u \rceil)\mathbbm{1}_{\{\varepsilon_u>0\}},\\
\mathbb{P}[(k_u,k_v) = (\lceil z_u \rceil, \lceil z_v \rceil)] &= 0.
\end{aligned}
\]

\paragraph{Case 2: $\varepsilon_u + \varepsilon_v > 1$.}
\[
\begin{aligned}
\mathbb{P}[(k_u,k_v) = (\lfloor z_u \rfloor, \lfloor z_v \rfloor)] &= 0,\\
\mathbb{P}[(k_u,k_v) = (\lfloor z_u \rfloor, \lceil z_v \rceil)] &= \mathbb{P}(z'_u = \lfloor z_u \rfloor),\\
\mathbb{P}[(k_u,k_v) = (\lceil z_u \rceil, \lfloor z_v \rfloor)] &= \mathbb{P}(z'_v = \lfloor z_v \rfloor),\\
\mathbb{P}[(k_u,k_v) = (\lceil z_u \rceil, \lceil z_v \rceil)] &= \mathbb{P}(z'_u = \lceil z_u \rceil) + \mathbb{P}(z'_v = \lceil z_v \rceil) - 1.
\end{aligned}
\]

We then define $z'_{u^*}$ for the fused cluster as $k_u + k_v$ and couple the distributions of the children accordingly. All other steps follow exactly as in \cite{bansal15polylog}.
\end{proof}

Combining \cref{corr: potential server}, \cref{lem:rounding-merge}, and the same reasoning as in \cref{thm: reduction} gives \cref{thm: k-server}.

\end{document}